\documentclass[a4paper, 10pt, pdflatex]{article}
\usepackage{amsfonts,amssymb,amsmath}
\usepackage{epsfig,epsf}
\usepackage{graphicx}
\usepackage{a4wide}
\usepackage{amsmath}
\usepackage{mathtools}
\usepackage{bbm}
\usepackage{amsthm}
\usepackage[toc,page]{appendix}
\usepackage{afterpage}
\usepackage{enumitem}
\usepackage{cleveref}
\setlist[enumerate,1]{label=\arabic*.,ref=\arabic*}
\setlist[enumerate,2]{label=\alph*.,ref=\arabic{enumi}.\alph*}
\setlist[enumerate,3]{label=\roman*.,ref=\arabic{enumi}.\alph{enumii}.\roman*}

\numberwithin{equation}{section}

\allowdisplaybreaks

\numberwithin{figure}{section}
\newtheorem{thm}{Theorem}[section]

\newtheorem{lem}[thm]{Lemma}
\newtheorem{prop}[thm]{Proposition}

\theoremstyle{definition}

\newtheorem{algorithm}[thm]{Algorithm}
\newtheorem{assumption}[thm]{Assumption}
\newtheorem{remark}[thm]{Remark}
\newtheorem{example}[thm]{Example}

\newcommand{\norm}[1]{\left\Vert#1\right\Vert}
\newcommand{\abs}[1]{\left\vert#1\right\vert}

\newcommand{\Real}{\mathbb R}

\newcommand*{\ud}{\mathrm{d}}
\newcommand{\pspace}{(\Omega,\mathcal{F},\mathbb{P})}
\newcommand{\pp}{\mathbb{P}}
\def\e{{\mathrm{e}}}

\begin{document}

\title{A Kalman particle filter for online parameter estimation\\ with applications to affine models}

\author{Jian He \and Asma Khedher \and Peter Spreij}

\date{} 


\maketitle

\begin{abstract}

\noindent  In this paper we address the problem of estimating the posterior distribution of the static parameters of a continuous time state space model with discrete time observations by an algorithm that combines the Kalman filter and a particle filter. The proposed algorithm is semi-recursive and has a two layer structure, in which the outer layer provides the estimation of the posterior distribution of the unknown parameters and the inner layer provides the estimation of the posterior distribution of the state variables. This algorithm has a similar structure as the so-called recursive nested particle filter, but unlike the latter filter, in which both layers use a particle filter, this proposed algorithm introduces a dynamic kernel to sample the parameter particles in the outer layer to obtain a higher convergence speed. Moreover, this algorithm also implements the Kalman filter in the inner layer to reduce the computational time. This algorithm can also be used to estimate the parameters that suddenly change value. We prove that, for a state space model with a certain structure, the estimated posterior distribution of the unknown parameters and the state variables converge to the actual distribution in $L_p$ with  rate of order $\mathcal{O}(N^{-\frac{1}{2}}+\delta^{\frac{1}{2}})$, where $N$ is the number of particles for the parameters in the outer layer and $\delta$ is the maximum time step between two consecutive observations. We present numerical results of the implementation of this algorithm, in particularly we implement this algorithm for affine interest models, possibly with stochastic volatility, although the algorithm can be applied to a  much broader class of models.
\smallskip\\
{\sl keywords:} affine process, state space model, Kalman filter, particle filter, parameter estimation, posterior distribution

\end{abstract}

\section{Introduction}

We pose the problem, describe its background and give a brief sketch of earlier approaches. After that we explain our approach and contribution to the literature and outline the organization of the present paper.

\subsection{Problem description and background}

When using stochastic models in a business environment, the model parameters need to be estimated, which turns out to be a very challenging problem. The main methods for parameter estimation can be classified into two groups: Bayesian and Maximum Likelihood estimation (MLE) methods. Such methods can also be categorized as online or offline depending on whether the data are used sequentially, or used in batches of observations. The MLE approach is to find the estimate which maximizes the marginal likelihood of the observed data. The Bayesian approach, however, considers the parameters as random variables which are updated recursively using prior knowledge of the parameters and the likelihood of the observations. In applications, the MLE based offline method is often linked to the Kalman filter or its modifications such as the extended Kalman filter, see~\cite{einicke1999robust, wan2001dual}, or the unscented Kalman filter, see~\cite{wan2001unscented},
because these algorithms can compute or approximate the likelihood function analytically. However, a common problem of the MLE calibration is that the likelihood function is usually not convex. Hence the numerical optimization of the likelihood often ends up at a local maximum instead of the global maximum. This problem can be even more severe when dealing with models with many parameters, such as multi-factor Hull-White models, popular in interest rate modeling. Moreover, the MLE method normally requires  static model parameters, while in reality the model parameters, such as volatility in financial models, could change over time. These issues restrict the application of offline methods, for instance in the financial modeling area. So, in recent decades, online methods received more and more attention.

Attempts to solve the problem of estimating the static parameters online was to include simulations (particles) of parameter values. One then has a \emph{particle filter}, see for example~\cite{doucet2000sequential,gordon1993novel,kitagawa1996,LiuChen} and~\cite{kantas2015} for a survey.
However,  through successive time steps this approach can quickly lead to what is called particle degeneracy of the parameter space. One solution to this degeneracy problem is to use a kernel density to estimate the posterior distribution of the parameters from which new parameter particles can be drawn at each time step~\cite{LW}. However, such a method can only work on some models with parameters of  low dimension. Also a convergence analysis of such a method is missing.

In recent years, some new methods have been proposed to deal with the online parameter estimation problem, including the iterated batch importance sampling (IBIS), see~\cite{chopin2002}, the sequential Monte Carlo square (SMC2) simulation, see~\cite{CJP}, and the recursive nested particle filter (RNP filter, also RNPF in short), see~\cite{CrisanMiguez}. The SMC2 and the RNPF use two layers of Monte Carlo methods to overcome certain difficulties with the IBIS method, see~\cite{papavasiliou}. An important difference between SMC2 and the RNPF is that the SMC2 is a non-recursive method, whereas the RNPF is recursive. Hence in general, RNPF is more efficient than SMC2.

In~\cite{CrisanMiguez} the estimated posterior measure of the parameters by using an RNPF algorithm is shown to converge to the actual measure in $L_p$-norm with rate $N^{-1/2}+M^{-1/2}$, where $N$ is the number of particles for the parameter estimation in the outer layer and $N\times M$ is the number of particles for the state variables in the inner layer. The RNPF has some drawbacks for a practical application. One is that the computation of  the two Monte Carlo layers is very time consuming, another one is that the RNPF requires that the parameter mutation size is small enough. As a consequence the RNPF converges very slowly to the actual value of the parameters and hence requires a very long time series of data, which is very often not available in many applications.

\subsection{Contribution}

In this paper,  we consider joint parameter and state estimation for a state space model where the state evolves continuously in time, whereas the observations are made at discrete time instants. We use a Bayesian online approach to parameter estimation. We propose an algorithm which combines the Kalman filter and a particle filter for online estimation of the posterior distribution of the unknown parameters. This algorithm has a similar structure as the RNPF, it is a semi-recursive algorithm with also two layer structure:  the inner layer provides the approximation on the posterior distribution of the state variables conditioned on the parameter particles generated in the outer layer, while the outer layer provides an approximation of the posterior distribution of the parameters by using the outcome of the inner layer.

Our proposed methodology has two main differences when compared to the RNPF algorithm. One difference is that in the inner layer, the posterior distribution of the state variables is estimated by the Kalman filter instead of a particle filter. The implementation of the Kalman filter reduces the computation complexity and hence results in a much faster and robust algorithm. The second difference is in the outer layer. In the RNPF the parameter samples are generated from a certain kernel function. In order to obtain a recursive algorithm, some requirements on the kernel function are introduced. This results in a kernel that significantly reduces the convergence speed of the RNPF. We overcome this problem by using dynamic jittering kernels. Especially in this paper, we implement two different kernel functions. One is applied at the beginning stage to obtain a higher convergence speed. The consequence, however, is that the algorithm is not recursive at this beginning stage since this kernel function does not satisfy the requirements of a recursive algorithm. The other kernel is applied when the variance of the parameter particles decreases to a certain level which is  such that this kernel function satisfies the conditions for a recursive method. From that time on, the algorithm is truly recursive. From the numerical experiments we performed, we observe that the variance of the particles decreases very fast at the beginning stage, usually after hundreds steps. Hence by using these two different kernel functions, the algorithm converges much faster than the RNPF.

This paper also provides theoretical results on the asymptotical behavior of the proposed algorithm. When dealing with non-Gaussian or non-linear models, the Kalman filter in the inner layer could produce a biased estimate of the posterior distribution of the state variables. This makes it difficult to generally study the convergence of the posterior distribution of the parameters. Although it is shown in \cite{Sara} that, under certain assumptions, the bias introduced in the inner layer makes the posterior distribution of the parameters converge to a biased distribution, this bias is intractable in general. In this paper, for models with a certain structure, we show that the estimated distributions of the parameters and the states converges the actual distributions in $L_p$ with  rate of order $\mathcal{O}(N^{-1/2}+\delta^{1/2})$ under certain regularity assumptions, where $N$ is the number of particles for the parameter space and $\delta$ is the maximum time step between consecutive observations. Note that we don't have to deal with particles in the inner layer, which improves on the order $M^{-1/2}$ term for convergence rate of the RNPF.
Our proofs are inspired by those in \cite{CrisanMiguez}, but at crucial steps we obtain novel results. These are due to the use of the Kalman filter in one of the layers and to  the size of  the  time discretization that governs the observations of the continuous time system, the latter not playing a role in the  setting of the cited reference.

To illustrate the performance of the algorithm, we present numerical results of the parameter estimation on several affine interest rate models, some allowing for stochastic volatility, including two-factor Hull-White model and the Cox-Ingersoll-Ross (CIR) model. For the CIR model we have also implemented the RNPF and we observed that our algorithm outperforms  the RNPF. Although the algorithm is designed for static parameter estimation,  it can also be used to estimate parameters that perform sudden changes in value. We also present an implementation of the algorithm in such a situation, and we observe that the algorithm is able to quickly track such a sudden change.

\subsection{Organization of the paper}
In Section~\ref{set-up} we present the state space model of interest. This section also provides brief reviews on  Bayesian filters, including the  Kalman filter and the particle filter, and online parameter estimation using particle filters.
Section~\ref{section:affinemodel} contains an encompassing framework for various affine models that are used in interest rate modeling and to which we apply our
proposed Kalman particle algorithm, which is introduced in Section~\ref{Kalman-particle-filter}. In Section~\ref{convergence_analysis} we provide the convergence analysis and in Section~\ref{Numerical results} the numerical results are presented. Finally Section~\ref{conclusion} is devoted to the conclusions. In the Appendix we collect some background results on affine processes.

\subsection{Notation}

Let $d\geq 1$, $S \subseteq \mathbb{R}^d$, and $\mathcal{B}(S)$ be the sigma algebra of Borel subsets of $S$.
We denote by $\mathbf{1}_A$ the indicator function on $A \in \mathcal{B}(S)$ and by $\delta_x$ the Dirac measure for a given $x \in S$, i.e.,
							\[
							\delta_{x}(A) =   \mathbf{1}_A(x)= \begin{cases}
																				1,  &  \mbox{if} \ x \in A\,,\\
																				0,  &  \mbox{otherwise}\,.
																		\end{cases}
						  \]
Suppose given a function $f: S\rightarrow \mathbb{R}$ and a probability measure $\mu$ on $(S, \mathcal{B}(S))$. We denote the integral of $f$ w.r.t.~$\mu$ by $(f,\mu) := \int_S f(x)\,\mu(\ud x)$ and the supremum norm of $f$ by $\| f\|_\infty = \sup_{x \in S} |f(x)|$.\\
We use the notation $x_{0:k} := (x_{0},\ldots,x_k)$ for a discrete-time sequence up to time $k$ of a process $(x_k)_{k \in \mathbb{N}}$.  By
$\cdot^\top$, we denote the transpose of a vector or a matrix. The Euclidian norm of an element $x \in \mathbb{R}^d$, is denoted by $\| x\|$ and
the $L^p$-norm, for $p\geq 1$ of a random variable $X$, defined on some probability space $\pspace$, is denoted by $\|X\|_p = (\mathbb{E}|X|^p)^{1/p}$.	Densities of random variables or vectors $x$ (always assumed to exist w.r.t.\ the Lebesgue measure) are often denoted $p$, or $p(x)$ and conditional densities of $X$ given $Y=y$ are often denoted $p(x\mid y)$, possibly endowed with sub- or superscripts.

\section{Set up and background on parameters estimation using filters}\label{set-up}

In this section we outline the set up, we pose the problem formulation, give a brief survey of various filters (Bayesian, Kalman, particle filter) and address the parameter estimation problem using particle filters. Time is assumed to be discrete.

\subsection{Discrete-time state space model}
We consider the following general {\it state space} model
\begin{equation}\label{state_space_model}
\begin{aligned}
x_k&=f_k(x_{k-1},u_k)\,\quad k \in \mathbb{N}^+,\\
y_k  &=h_k(x_k,v_k)\,,  \quad k \in \mathbb{N}^+\,,
\end{aligned}
\end{equation}
where
$f_k:\Real^d \times \Real^d \rightarrow \Real^d$, $h_k: \Real^d \times \Real^d \rightarrow \Real^m$ are given functions and $\{u_k\}_{k\in \mathbb{N}^+}$
and $\{v_k\}_{k\in \mathbb{N}^+}$ are $d$-dimensional white noise processes, possibly independent, and both independent of the initial condition $x_0$, all defined on some $\pspace$. Parameters in the functions $f_k$ and $h_k$, together with the covariance of $u_k$ and $v_k$ can be seen as the parameters of the state space model, and to which we refer to as $\theta$.\\
\\
\noindent It follows that the model \eqref{state_space_model} satisfies the properties of a stochastic system, i.e.\ at every (present) time $k\geq 1$ the future states  and future observations $(x_j,y_j)$, $j\geq k$, are conditionally independent from the past states and observations $(x_j,y_{j-1})$, $j\leq k$,  given the present state $x_k$, see~\cite{vanschuppen1989}.
It then follows that  $\{x_k\}_{k \in \mathbb{N}}$ is a Markov process,
and for every $k\geq 1$ one has that $y_k$ and $y_{1:k-1}$ are conditionally independent given $x_k$, in terms of densities,
 \begin{align}\label{eq:cond-ind-y}
 p(y_k\mid y_{1:k-1},x_k)=p(y_k\mid x_k)\,,\quad  \mbox{for}\quad  k \in \mathbb{N}^+\,.
 \end{align}
Moreover, one also has, for every $k\geq 1$, that $x_k$ and $y_{1:k-1}$ are conditionally independent given $x_k$, in terms of densities,
\begin{equation}\label{eq:xyk}
p(x_k\mid x_{k-1},y_{1:k-1})=p(x_k\mid x_{k-1}).
\end{equation}
The latter equation has the consequence
\begin{equation}\label{eq:xyk-1}
p(x_k\mid y_{1:k-1})=\int p(x_k\mid x_{k-1})p(x_{k-1}\mid y_{1:k-1})\ud x_{k-1}.
\end{equation}
\noindent We are interested in estimating the (latent) state process $\{x_k\}_{k \in \mathbb{N}^+}$, but only have access to the process $\{y_k\}_{k \in \mathbb{N}^+}$ which represents the observations. Because of the existence of the white noise in the data, estimating the value of the latent states $\{x_k\}_{k \in \mathbb{N}^+}$ by the observations $\{y_k\}_{k \in \mathbb{N}^+}$ is not trivial. There are different methodologies in the literature to estimate the latent process (see e.g.~\cite{Press, ChuiChen, arulampalam2002tutorial}). We introduce some of these methodologies in our paper since we will need them in our analysis later. We first introduce the Bayesian filter.

\subsection{Bayesian filter of discrete-time Markovian state space model}

\noindent The {\it Bayesian filter}, see e.g.~\cite{Press,Robert} for an overview, is used to estimate the latent states
$\{x_k\}_{k \in \mathbb{N}^+}$ in \eqref{state_space_model} given the parameter $\theta$.
We define the initial probability measure $\pi_{0}$ of $x_{0}$, and the transition measure $\pi_k^{\theta}$ of $x_k$ under a given parameter $\theta$ at time $k$ by
\begin{equation}
\begin{aligned}
\pi_{0}(A)                      &= \mathbb{P}(x_{0}\in A),\\
\pi_k^{\theta}(A\mid x_{k-1}) &= \mathbb{P}(x_k\in A\mid x_{k-1},\theta), \qquad k\in \mathbb{N}^+\,,
\end{aligned}
\end{equation}
where $A\in \mathcal{B}(\mathbb{R}^{d})$ is a Borel set.

The methodology in Bayesian filtering consists of two parts: prediction and update. At every time point $k$, the prediction part computes (estimates) the {\it prior measure} of $x_k$ (a time $k$ given the past observations up to time $k-1$) and the update part computes (estimates) the {\it posterior measure} of $x_k$ given the past up to time $k$, respectively given by
\begin{equation}\label{prior-posterior}
\begin{aligned}\gamma_k^{\theta}(\ud x_k) &= \pp(\ud x_k\mid y_{1:k-1},\theta)\,, \\
\Gamma_k^{\theta}(\ud x_k) &= \pp(\ud x_k\mid y_{1:k}, \theta)\,, \qquad k\in \mathbb{N}^+\,.
\end{aligned}
\end{equation}
\noindent   Using Bayes' rule, we deduce that the density function of the prior distribution is given by
\begin{equation*}
\begin{aligned}
p(x_k\mid y_{1:k-1}, \theta) &= \int p(x_k\mid x_{k-1},y_{1:k-1}, \theta)p(x_{k-1}\mid y_{1:k-1}, \theta)\, \ud x_{k-1}\\
                           &= \int p(x_k\mid x_{k-1}, \theta)p(x_{k-1}\mid y_{1:k-1}, \theta)\, \ud x_{k-1}\,,
\end{aligned}
\end{equation*}
where we used \eqref{eq:xyk} to get the last equality. This
implies the relation
\begin{equation}\label{eq:gpG}
\gamma_k^\theta(\ud x_k) =\int \pi_k^\theta(\ud x_k\mid x_{k-1})\Gamma_{k-1}^\theta(\ud x_{k-1}).
\end{equation}
Let $g: \mathbb{R}^d\rightarrow \mathbb{R}$ be an integrable function w.r.t.~the measure $\gamma_k^\theta$. Then we get by Fubini's theorem
\begin{align*}
(g,\gamma_k^{\theta}) &= \int\int g(x_k)\pi_k^{\theta}(\ud x_k\mid x_{k-1})\Gamma_{k-1}^{\theta}(\ud x_{k-1})\\
& = \int (g,\pi^\theta(\cdot\,\mid x_{k-1}))\Gamma_{k-1}^{\theta}(\ud x_{k-1}),
\end{align*}
which we abbreviate by
\begin{equation}\label{int_gamma}
(g,\gamma_k^{\theta})= ((g,\pi_k^{\theta}),\Gamma_{k-1}^{\theta}).
\end{equation}

\noindent The purpose of the Bayesian algorithm is to sequentially compute the posterior measure $\Gamma_k^{\theta}$. Let
\begin{equation*}
l^{\theta}_{y_k}(x)=p(y_k\mid x,\theta)
\end{equation*}
be the probability (with some abuse of statistical terminology we often also call it likelihood) of the realized observation $y_k$ conditional on the state value $x_k=x$ and the model parameter $\theta$.
Then using Bayes' rule, \eqref{eq:cond-ind-y} and \eqref{eq:gpG}, we obtain for a function $g$ that is integrable w.r.t.\ $\Gamma_k^{\theta}$
\begin{align*}
(g,\Gamma_k^{\theta})&= \int g(x_k)p(x_k\mid y_{1:k}, \theta)\, \ud x_k\\
					 &= \int g(x_k) \frac{p(x_k,y_k,y_{1:k-1}\mid  \theta)}{p(y_{1:k}\mid  \theta)}\,\ud x_k\\
					 &= \frac{\int g(x_k)p(y_k\mid x_k,y_{1:k-1}, \theta)p(x_k\mid y_{1:k-1}, \theta)\, \ud x_k}{p(y_k\mid y_{1:k-1}, \theta)}\\
					 &= \frac{\int g(x_k)p(y_k\mid x_k, \theta)p(x_k\mid y_{1:k-1}, \theta)\, \ud x_k}{\int p(y_k\mid x_k, \theta)p(x_k\mid y_{1:k-1}, \theta)\, \ud x_k}\\
					 &= \frac{\int g(x_k)l^\theta_{y_k}(x_k)p(x_k\mid y_{1:k-1}, \theta)\, \ud x_k}{\int l^\theta_{y_k}(x_k)p(x_k\mid y_{1:k-1}, \theta)\, \ud x_k}\\
					 &= \frac{\int g(x_k)l^\theta_{y_k}(x_k)\int \pi_k^\theta(\ud x_k\mid x_{k-1})\Gamma_{k-1}^\theta(\ud x_{k-1})}{\int l^\theta_{y_k}(x_k)\int \pi_k^\theta(\ud x_k\mid x_{k-1})\Gamma_{k-1}^\theta(\ud x_{k-1})},
\end{align*}
which we abbreviate, similar to \eqref{int_gamma}, by
\begin{equation}\label{posterior measure update}
(g,\Gamma_k^{\theta})= \frac{((gl^{\theta}_{y_k},\pi_{{k}}^{\theta}),\Gamma_{k-1}^{\theta})}{((l^{\theta}_{y_k},\pi_{{k}}^{\theta}),\Gamma_{k-1}^{\theta})}.
\end{equation}
 If we assume the likelihood function $l^{\theta}_{y_k}$ and the transition measure $\pi_{{k}}^{\theta}$ are known, then given the posterior measure $\Gamma_{k-1}^{\theta}$, we can use Equation (\ref{posterior measure update}) to compute the posterior measure $\Gamma_{k}^{\theta}$. In this way the posterior measure $\{\Gamma_{k}^{\theta}\}_{k\in \mathbb{N}^+}$ can be computed recursively. Moreover, using \eqref{eq:cond-ind-y} again, the conditional likelihood $p(y_k\mid y_{1:k-1},\theta)$ and the likelihood $p(y_{1:k}\mid \theta)$ can be respectively computed as
\begin{align}\label{MarginalLikelihood}
p(y_k\mid y_{1:k-1},\theta) &= \int p(y_k\mid x_{k},y_{1:k-1},\theta)p(x_k\mid y_{1:k-1},\theta)\, \ud x_k\nonumber\\
														&= \int p(y_k\mid x_k,\theta)p(x_k\mid y_{1:k-1},\theta)\, \ud x_k\nonumber\\
														&= (l^{\theta}_{y_k},\gamma_{k}^{\theta}),
\end{align}
and
\begin{align*}
p(y_{1:k}\mid\theta)        &= P(y_{1}\mid\theta)\prod_{i=2}^k  p(y_{t_i}\mid y_{1: i-1},\theta) \\
														&= P(y_{1}\mid\theta)\prod_{i=2}^k (l^{\theta}_{y_{i}},\gamma_{i}^{\theta})\,.
\end{align*}

\noindent When \eqref{state_space_model} is a  linear Gaussian model, then the Bayesian filter is equivalent to the Kalman filter, which we briefly review in the next subsection.

\subsection{Kalman filter}\label{Kalman-filter}

We assume that the state and observations in \eqref{state_space_model} evolve according to a linear Gaussian model. That is the functions $f_k$ and $h_k$ have to take linear forms as follows
\begin{equation}
\begin{aligned}
x_k&= F_k x_{k-1} + u_k\,,\\
y_k  &= H_k x_k + v_k\,, \qquad k\in \mathbb{N}^+\,,
\end{aligned}
\label{KF}
\end{equation}
where $F_k$ is a $d\times d$ matrix, $H_k$ is a $m\times d$ matrix and the noise terms $u_k$ ($d$-dimensional), $v_k$ ($m$-dimensional) are assumed to be Gaussian with mean $0$ and
variance $Q_k$, $R_k$, respectively. Moreover, the initial state $x_{0}$ is assumed to be Gaussian.
Due to the Gaussian assumptions and the linear structure of the model in \eqref{KF}, one can derive analytic expressions for the prior and posterior measures defined in \eqref{prior-posterior} and the algorithm in the {\it Kalman filter}, see e.g.~\cite{ChuiChen,GrewalAndrews}, yields the exact solution to the estimation problem.

Denote by $N(\ud x;\mu,\Sigma)$ or $N(\mu,\Sigma)$ the Gaussian distribution with mean $\mu$ and Covariance $\Sigma$. We also use the generic notation $N(x;\mu,\Sigma)$ to denote the density at $x$ of this normal distribution.
Recall from \eqref{prior-posterior}, the prior and posterior measures and
denote by $B_{k-1}$ and $P_{k-1}$ respectively, the mean and the covariance of the posterior measure at time ${k-1}$.
Then the prior measure is given by
\begin{equation*}
\begin{aligned}
\gamma_k^{\theta} (\ud x_k) = N(\ud x_k; F_k\bar{x}_{k-1}, F_kP_{k-1}F_k^\top+Q_k)\,,
\end{aligned}
\end{equation*}
which implies that the prior measure is a conditionally Gaussian measure with mean and covariance respectively given by
$$A_k^- = F_kB_{k-1}\,, \qquad P_k^- = F_kP_{k-1}F_k^\top+Q_k\,.$$
Moreover, the posterior measure is given by
\begin{align}\label{eq:posterior-measure}
&\Gamma_k^{\theta} (\ud x_k)
= N(\ud x_k; B_k, C_k)\,,
\end{align}
where
$$B_k = A_k^- + P_{{k}}^- H_k^\top (H_kP_{{k}}^-H_k^\top + R_k)^{-1}(y_k-H_kA_k^-)\,,$$
$$C_k = P_{{k}}^- - P_{{k}}^- H_k^\top (H_kP_{{k}}^-H_k^\top + R_k)^{-1}H_kP_{{k}}^-\,.$$
Finally, the conditional likelihood is given by
\begin{equation}\label{MarginalLikelihood_KF}
p(y_k\mid y_{1:k-1},\theta) = N(y_k; H_kA_k^-,H_kP_{{k}}^-H_k^\top + R_k)\,.
\end{equation}
Let $S_k = H_kP_{{k}}^-H_k^\top + R_k$, $k\in \mathbb{N}^+$. Then we obtain the recursion for the log-likelihood of the observation $\log(p(y_{1:k})\mid \theta)$ as follows,
\begin{align*}
\log (p(y_{1:k})\mid \theta)&=\log p((y_{1:k-1})\mid \theta)\\
&\qquad - \frac{1}{2}\left(m\log 2\pi - \log(\det(S_k)) - (y_k-H_k\bar{x}_{{k}}^- )^\top S_k^{-1} (y_k-H_k\bar{x}_{{k}}^- ) \right)\,,
\end{align*}
where $m$ is the dimensionality of the data $y_k$, $k\in \mathbb{N}^+$. Hence, by maximizing the likelihood of the observations, one can determine the optimal parameters of the linear Gaussian system \eqref{KF}.\\
For most non linear non Gaussian models, it is not possible to compute the prior and posterior measures analytically and numerical methods are called for. In this case, the particle filter, which we introduce in the next subsection, is widely used.

\subsection{Particle filter}\label{particle-filter}

In the {\it particle filter}, see e.g.~\cite{arulampalam2002tutorial,cappe2007overview,doucet2009tutorial}, the prior and posterior distributions are estimated by a Monte Carlo method. With a Monte Carlo method, a certain measure $\mu$ is generally estimated by \begin{equation*}
\mu^N(\ud x)= \sum_{i=1}^{N}a^{(i)}\delta_{x^{(i)}}(\ud x)\,,
\end{equation*}
where $\{x^{(i)}, i=1,\cdots,N\}$ are i.i.d.~random samples from a so-called {\it importance density} and $\{a^{(i)}, i=1,\cdots,N\}$ are {\it the importance weights}. The key part of the particle filter is to choose the importance density and compute the importance weights, see e.g.~\cite{D}. For the general state space model (\ref{state_space_model}), suppose the posterior measure $\Gamma_{k-1}^{\theta}$ at time ${k-1}$ is estimated by
\begin{equation*}
{\Gamma}_{k-1}^{\theta}\approx \sum_{i=1}^{N}a_{k-1}^{(i)}\delta_{x_{k-1}^{(i)}}.
\end{equation*}
If at time $k$, the samples $\tilde{x}_k^{(i)}$ are generated from the transition measure $\pi_k^{\theta}(\ud x\mid x_{k-1}^{(i)})$ for $i=1,\cdots,N$, then using Equation (\ref{int_gamma}), the integral $(f,\gamma_k^{\theta})$ can be estimated by
\begin{equation*}
(f,\gamma_k^{\theta})\approx \sum_{i=1}^N a_{k-1}^{(i)}f(\tilde{x}_k^{(i)}).
\end{equation*}
Moreover, using Equation \eqref{posterior measure update},
the prior and posterior measures are respectively estimated by
\begin{equation}\label{eq:pp-measures}
\begin{aligned}
\gamma_k^{\theta} &\approx \sum_{i=1}^N a_{k-1}^{(i)}\delta_{\tilde{x}_k^{(i)}},\\
\Gamma_k^{\theta} &\approx \sum_{i=1}^N a_{k}^{(i)}\delta_{\tilde{x}_k^{(i)}}
\end{aligned}
\end{equation}
and from \eqref{MarginalLikelihood}, we deduce the following approximation for the conditional likelihood
\begin{equation}\label{eq:cond-likelihood}
p(y_k\mid y_{1:k-1}, \theta)\approx \sum_{i=1}^N a_{k-1}^{(i)}l^{\theta}_{y_k}(\tilde{x}_k^{(i)})\,.
\end{equation}
Consequently, the integral $(f,\Gamma_k^{\theta})$ can be estimated by
\begin{equation*}
\begin{aligned}
(f,\Gamma_k^{\theta})&\approx \frac{\sum_{i=1}^N a_{k-1}^{(i)}l^{\theta}_{y_k}(\tilde{x}_k^{(i)})f(\tilde{x}_k^{(i)})}{\sum_{i=1}^N a_{k-1}^{(i)}l^{\theta}_{y_k}(\tilde{x}_k^{(i)})}\\
                         &= \sum_{i=1}^N a_{k}^{(i)}f(\tilde{x}_k^{(i)})\,,
\end{aligned}
\end{equation*}
where the weights $a_{k}^{(i)}$ are defined by
\begin{equation}\label{eq:ak}
a_{k}^{(i)} = \frac{a_{k-1}^{(i)}l^{\theta}_{y_k}(\tilde{x}_k^{(i)})}{\sum_{i=1}^N a_{k-1}^{(i)}l^{\theta}_{y_k}(\tilde{x}_k^{(i)})}.
\end{equation}
\noindent Equations \eqref{eq:pp-measures} and \eqref{eq:cond-likelihood} show how to sequentially estimate the posterior measure $\Gamma_k^{\theta}$ using the Monte Carlo method. This type of particle filter is often referred to as {\it sequential particle filter}. In \cite{D} it is shown that the variance of the importance weights decreases stochastically over time. This will lead the importance weights to be concentrated on a small amount of sampled particles. This problem is called \emph{degeneracy}.
To address the rapid degeneracy problem, the \emph{sampling-importance resampling (SIR)} method, see e.g.~\cite{D, PS}, is introduced to eliminate the samples with low importance weight and multiply the samples with high importance weight. In SIR, once the approximation
of the posterior measure $\Gamma_k^{\theta}\approx \sum_{i=1}^N a_{k}^{(i)}\delta_{\tilde{x}_k^{(i)}}$ is obtained, new, re-sampled,  particles ${x}_k^{(j)}$ are i.i.d.\ sampled from this approximated measure, i.e.\ every $x_k^{(j)}$ is independently chosen from the $\tilde{x}_k^{(i)}$ with probabilities $a_k^{(i)}$, for $i=1,\cdots,N$. This step can be accomplished by sampling integers $j$ from $\{1,\ldots,n\}$ with probabilities $a_{k}^{(i)},i=1,\cdots,N$. Then the new estimation on the posterior measure $\Gamma_k^{\theta}$ is given by
\begin{equation*}
\Gamma_k^{\theta} \approx \frac{1}{N}\sum_{i=1}^N\delta_{{x}_k^{(i)}}
\end{equation*}
and the new estimate of the conditional likelihood is
\begin{equation*}
p(y_k\mid y_{1:k-1}, \theta)\approx \frac{1}{N}\sum_{i=1}^N l^{\theta}_{y_k}({x}_k^{(i)})\,.
\end{equation*}

\subsection{Static model parameters estimation using particle filter}\label{sec:static-model}

\noindent When the parameters are known, the particle filter is a quite effective algorithm for latent variable estimation. However, if the parameters are not known beforehand, it is a very challenging task to estimate the parameters and the latent states using the particle filter. Here we take a Bayesian approach to estimate the parameters.  The estimation of the parameters in online
estimation requires the computation of the posterior distribution of $\theta$, i.e., $p(\theta \mid y_{1:k}), k \in \mathbb{N}^+$. Using Bayes' rule, one can represent the posterior density as
\begin{equation*}
p(\theta\mid y_{1:k}) = \frac{p(y_k\mid y_{1:k-1},\theta)p(\theta\mid y_{1:k-1})}{\int p(y_k\mid y_{1:k-1},\theta)p(\theta\mid y_{1:k-1}) \,\ud\theta}\,.
\end{equation*}

\noindent Hence, the posterior distribution of $\theta$ given $y_{1:k}$ can be evaluated as

\begin{equation*}
\mathbb{P}(\theta\in A \mid y_{1:k}) = \frac{\int \mathbf{1}_A(\theta) p(y_k\mid y_{1:k-1},\theta)p(\theta\mid y_{1:k-1})\,\ud\theta}{\int p(y_k\mid y_{1:k-1},\theta)p(\theta\mid y_{1:k-1})\, \ud\theta}\,.
\end{equation*}

\noindent To estimate the density and the distribution of $\theta$ given $y_{1:k}$, a straightforward way is to sample parameter particles from the former posterior distribution $p(\theta\mid y_{1:k-1})$. Denote the samples by $\{\theta^{(i)}, i=1,\cdots,N\}$, then the measure $p(\ud\theta\mid y_{1:k})$ at time $k$ can be approximated by
\begin{equation}\label{eq:posterior-theta}
\sum_{i=1}^N \frac{p(y_k\mid y_{1:k-1},\theta^{(i)})}{\sum_{i=1}^N p(y_k\mid y_{1:k-1},\theta^{(i)})} \delta_{\theta^{(i)}}(\ud\theta)=\sum_{i=1}^N w_k^{\theta^{(i)}}\delta_{\theta^{(i)}}(\ud\theta)\,,
\end{equation}
where the weights $w_k^{\theta^{(i)}}$, $i=1, \cdots N$, are defined by
\begin{equation}\label{theta-weights}
w_k^{\theta^{(i)}}=\frac{p(y_k\mid y_{1:k-1},\theta^{(i)})}{\sum_{i=1}^N p(y_k\mid y_{1:k-1},\theta^{(i)})}\,.
\end{equation}

\noindent There are two issues to implement \eqref{eq:posterior-theta}. One is that sampling from the former posterior distribution $p(\theta\mid y_{1:k-1})$ usually cannot be carried out exactly. Another is that often the likelihood $p(y_k\mid y_{1:k-1},\theta^{(i)})$ cannot be computed theoretically. These two latter issues can be tackled by using the {\it recursive nested particle filter (RNPF)}, recently introduced in \cite{CrisanMiguez}, which  is presented below.

\subsubsection{Recursive nested particle filter}\label{RNP-filter}

In the RNPF, a two layer Monte Carlo method is used. In the first layer, also referred to as outer layer, new parameter samples are generated by using a kernel function. This step is usually called {\it jittering} and the kernel is referred to as {\it the jittering kernel}. In the second layer, also called inner layer, a particle filter is applied to approximate the conditional likelihood $p(y_k\mid y_{1:k-1},\theta^{(i)})$. In the following paragraph of this section we present the RNPF in more detail and introduce its ensuing Algorithm~\ref{algoRNPF}.

First, assume that $\theta$ has a compact support $D_{\theta} \subset \mathbb{R}^{d_\theta}$, where $d_\theta$ is the dimension of $\theta$. Moreover assume at time ${k-1}$, one can generate a random grid of samples in the parameter space $D_{\theta}$, say $\{\theta_{k-1}^{(i)}, i=1,\cdots,N\}$, and for each $\theta_{k-1}^{(i)}$, we have the set of particles in the state space $\{x_{k-1}^{(i,j)}, 1\leq j\leq M\}$.
\begin{itemize}
\item {\bf Jittering.}
Given the parameters samples $\{\theta_{k-1}^{(i)}, i=1,\cdots,N\}$ at time $k-1$, new particles $\{\tilde{\theta}^{(i)}_k, i=1,\cdots,N\}$ at time $k$ are generated by some Markov kernels denoted by $\kappa(\theta\mid  \theta_{{k-1}}^{(i)})$ (step~\ref{1jitt1} in Algorithm~\ref{algoRNPF} below).
This step is the outer Monte Carlo layer.

\item{\bf Update.} From Equations \eqref{int_gamma} and \eqref{MarginalLikelihood}, we know that for a given $\tilde{\theta}$, the marginal likelihood is obtained by calculating the integral
\begin{align*}
p(y_k\mid y_{1:k-1},\tilde{\theta}) = ((l_{\tilde{\theta}}^{y_k}, \pi_k^{\tilde{\theta}}),\Gamma_{k-1}^{\tilde{\theta}})\,.
\end{align*}
In order to compute this latter integral, the posterior measure at time ${k-1}$, $\Gamma_{k-1}^{\tilde{\theta}}$, needs to be known. In the standard Bayesian filter, the parameters are fixed over time and this posterior measure is computed at time ${k-1}$ by using Equation (\ref{posterior measure update}). However in this case, this measure is not directly available since the parameter has evolved from $\theta$ at time ${k-1}$ to $\tilde{\theta}$ at time $k$. In order to compute $\Gamma_{k-1}^{\tilde{\theta}}$, one needs to re-run a filter from time $1$ to $k$, which makes the algorithm not recursive and very time consuming.
 The authors in \cite{CrisanMiguez} solved this latter problem by assuming that $\Gamma_{k-1}^{\theta}$ is continuous w.r.t.\
$\theta \in D_\theta$, which means that when $\theta \approx \tilde{\theta}$, then $\Gamma_{k-1}^{\theta} \approx \Gamma_{k-1}^{\tilde{\theta}}$. Therefore by considering a rather small variance in the jittering kernel, one can use the particle
approximation of the filter computed for $\theta$ at time $k-1$ as a particle approximation of the filter for the new sampled $\tilde{\theta}$ at time $k$.

In the RNPF, the jittering kernel is chosen such that the mutation step from $\theta_{k-1}$ to $\tilde{\theta}_k$ is sufficiently small, see Section 4.2 in \cite{CrisanMiguez}.
Then for each $\tilde{\theta}_k^{(i)}$, $\{i=1,\cdots, N\}$, a sequential nested particle filter (see Section~\ref{particle-filter} for the description of the particle filter methodology) is used for the state space to obtain $\{\tilde{x}_k^{(i,j)}, \,1\leq j\leq M\}$; see steps~\ref{1jitt2}, \ref{1jitt4},~\ref{1jitt5} in Algorithm~\ref{algoRNPF} below. This is the inner Monte Carlo layer.

\item{\bf Resampling.} The outer layer Monte Carlo method in the update step above provides an approximation of the likelihood $p(y_k\mid y_{1:k-1},\tilde{\theta}_k^{(i)})$, $i=1,\cdots, N$, (step~\ref{1jitt3} in Algorithm~\ref{algoRNPF}) which are used to re-weight the parameter particles and obtain $\{{\theta}_k^{(i)},{x}_k^{(i,j)},\, i=1\cdots, N,\, j=1,\cdots, M\}$; see step~\ref{1res} in Algorithm~\ref{algoRNPF}.
\end{itemize}
The RNPF is introduced in \cite{CrisanMiguez}. We reproduce it here for the sake of completeness.

\begin{algorithm}[{\bf sequential nested particle filter for parameter estimation}]\label{algoRNPF}
\mbox{}\vspace{-0.8em}\\
\begin{description}
\item[Initialization:]
			Assume an initial distribution $p(\theta_{0})$ for the parameters and $p(x_{0})$ for the states, and
				sample from the initial distributions to get $N$ particles $\{\theta_{0}^{(i)}, i=1,\cdots, N\}$ and $N\times M$ particles
				$\{x_{0}^{(i,j)}, i=1,\cdots,N, \,j=1,\cdots,M\}$.
\item [Recursion:] \mbox{}
\begin{enumerate}		
	\item Filtering: given $\{\theta_{k-1}^{(i)},x_{k-1}^{(i,j)}\}$, for each $i=1, \cdots, N$,
			\begin{enumerate}	
				\item (jittering, outer Monte Carlo layer) sample new parameters $\tilde{\theta}_k^{(i)}$ from the jittering kernel $\kappa(\theta\mid \theta_{{k-1}}^{(i)})$,\label{1jitt1}
				\item (together with the next two steps, this is the update part) sample new parameters $\hat{x}_k^{(i,j)}$, $j=1,\cdots, M$, from the transition measure $\pi_k^{\tilde{\theta}_k^{(i)}}(\ud x\mid x_{k-1}^{(i,j)})$ (inner Monte Carlo Layer),\label{1jitt2}
				\item compute $p(y_k\mid y_{1:k-1}, \tilde{\theta}_k^{(i)}) \approx \frac{1}{M} \sum_{j=1}^{M} l^{\tilde{\theta}_k^{(i)}}_{y_k}(\hat{x}_k^{(i,j)})$,\label{1jitt3}
				\item compute the weights for the state space using Equation~\eqref{eq:ak}\label{1jitt4}
				 $$a_{k}^{(j)} =\frac{ a_{k-1}^{(j)} p(y_{k}\mid \hat{x}_k^{(i,j)},\tilde{\theta}_k^{(i)})}{\sum_{j=1}^M a_{k-1}^{(j)} p(y_{k}\mid \hat{x}_k^{(i,j)},\tilde{\theta}_k^{(i)})}, \qquad j=1,\cdots, M\,,$$
				\item resample the $\hat{x}^{(i,p)}$: set $\tilde{x}^{(i,j)}$ equal to $\hat{x}^{(i,p)}$ with probability $a_{k}^{(p)}$, where $j,p \in \{1,\cdots M\}$.\label{1jitt5}
				
\end{enumerate}	
				\item\label{1res} Resampling of the $\{\theta_k^{(i)}\}$: compute the weights for the parameters space using Equation \eqref{theta-weights}
\begin{equation}\label{eq:wtilde}
w_k^{\tilde{\theta}^{(i)}_k}= \frac{p(y_k\mid y_{1:k-1},\tilde{\theta}_k^{(i)})}{\sum_{i=1}^N p(y_k\mid y_{1:k-1},\tilde{\theta}_k^{(i)})} , \qquad i=,1\cdots, N\,.
\end{equation}
                For each $i=1,\cdots,N$, set $\{\theta_k^{(i)}, x_k^{(i, j)}\}_{1\leq j\leq M}$ equal to  $\{\tilde{\theta}_k^{(p)}, \tilde{x}_k^{(p, j)}\}_{1\leq j\leq M}$ with probability
                $w_k^{\tilde{\theta}^{(p)}_k}$, where $p\in \{1,\cdots,N\}$.
                \item Go back to the filtering step.
\end{enumerate}
\end{description}
\end{algorithm}

\subsubsection{A note on the convergence of the RNPF}
A convergence study of Algorithm~\ref{algoRNPF} was carried out in Lemmas 3 to 6 and Theorems 2 and 3 in \cite{CrisanMiguez}, where the reasoning was split in the three steps of the algorithm: the jittering, the update and the resampling.
In this latter paper, it was proven that, under some regularity conditions, the $L^p$-norms of the approximation errors, induced by these different steps, vanish with rate proportional to $\frac{1}{\sqrt{N}}$ and $\frac{1}{\sqrt{M}}$. Recall here that $N$ and $N\times M$ are respectively the number of samples in the parameter space and the number of particles in the state space. A similar result was proven for the approximation of the joint posterior distribution of the parameters and the state variables. We will make use of some of these convergence results later in Section~\ref{convergence_analysis} to carry out convergence study of our proposed algorithm, Algorithm~\ref{algoKPF}.

Under the assumption that the posterior measure $\Gamma_k^\theta(\ud x)$ is continuous w.r.t.\ the parameter $\theta$ and when the mutation step of the parameters is small enough, the RNPF is a recursive algorithm. This makes the RNPF more efficient than non-recursive methods such as {\it sequential Monte Carlo square}, see~\cite{CJP}, and {\it Markov Chain Monte Carlo} methods, see~\cite{GHF,geweke1999markov,higdon1998auxiliary}.
The drawbacks of the RNPF are its heavy computational burden and slow convergence speed which are respectively due to the nested simulations in the two Monte Carlo layers and the small mutation step of the parameters. In many applications, such as in financial modeling, the time length of the data is quite limited. Hence the time series of the data are not enough to make the RNPF converge. To tackle this problem, we propose a new methodology in Section~\ref{Kalman-particle-filter}.

\section{Parameters estimation in short rate models}\label{section:affinemodel}

Here we present a rather general model, an affine process, particularly relevant in mathematical finance for instance where one is  interested in estimating the parameters of the short rate curve given the observed data. It motivates the kind of system that we will consider and to which the new (Kalman particle) filter of Section~\ref{Kalman-particle-filter} will be applied.

Let $(\Omega, \mathcal{F}, (\mathcal{F}_t)_{t\geq 0},\mathbb{P})$ be a filtered probability space satisfying the usual conditions and $(W_t)_{t\geq 0}$ be a $d$-dimensional Brownian motion.
In this paper, although our results can be applied to general state space models of type \eqref{state_space_model}, we will mainly consider dynamics of the type
\begin{equation}\label{model_analysis}
\ud x_t = A(\beta-x_t)\, \ud t + \left(\Sigma+\tilde{\Sigma}\sqrt{x_t^{(1)}}\right)\, \ud W_t, \,
x_0 = x \in \mathbb{R}^d\,,\\
\end{equation}
where $A,\Sigma$ and $\tilde{\Sigma}$ are $d\times d$-matrices, $\beta$ is a d-vector and its first component is non-negative,  and $(x_t^{(1)})_{t\geq 0}$ is the first component of $(x_t)_{t\geq 0}$.
We assume the matrix $A$ is diagonal and we denote the diagonal elements of $A$ by $\alpha_1,\cdots,\alpha_{d}$.
Consider some integers $p, q\geq 0$ with $p+q=d$. When $\Sigma\tilde{\Sigma}=\tilde{\Sigma}\Sigma=0$ and the parameters of the model \eqref{model_analysis} satisfy certain conditions known in the literature as {\it admissibility conditions}, the process
$(x_t)_{t\geq 0}$ is $\Real^p_+ \times \Real^q$-valued {\it affine} process,
see~\cite{DDW,DPS,KM} for an overview of affine processes. In Appendix~\ref{app-section:admissibility} we specify the admissibility of the parameters of the dynamics \eqref{model_analysis}.
In our context, the short rate evolution will be described by a process $(r_t)_{t\geq 0}$ given in terms of $(x_t)_{t\geq 0}$ by
$$r_t = c + \gamma^\top x_t\,,$$
where $c \in \mathbb{R}$, $\gamma \in \mathbb{R}^d$.
Let $T>0$ be the maturity time, then the {\it zero coupon bond} price at time $t<T$ is defined as
$$P(t,T) = \mathbb{E}[\e^{-\int_t^T r(s) \, \ud s}\mid \mathcal{F}_t]\,,$$
and the corresponding zero rates, also called yields, are defined as $-\log P(t,T)/(T-t)\,.$
The fact that the process $(x_t)_{t\geq 0}$ is affine, which happens if $\Sigma\tilde{\Sigma}$ is zero, allows one to obtain an explicit formula for the zero coupon bond price, i.e.
\begin{align}\label{eq:bond-explicit-formula}
P(t,T) =\e^{-\phi(T-t,0) -\psi(T-t,0) x(t)}\,.
\end{align}
The functions $\phi$ and $\psi$ are the solutions to some ordinary differential equations, which are often referred to as the Riccati equations, see Theorem~\ref{thm:discounting} in Appendix~\ref{sec:Riccati-equations} for details. Then,  if $\tilde{\Sigma}=0$ (first case), Equation~\eqref{eq:bond-explicit-formula} holds for $(\phi, \psi)$ the solution to \eqref{riccati-equations1}. If $\Sigma=0$ (second case), then \eqref{eq:bond-explicit-formula} holds for $(\phi, \psi)$ the solution to \eqref{riccati-equations2}. Denote the time to maturity $T-t$ by $\tau$, then the zero rate at time $t$ with time to maturity $\tau$ can be computed by
\begin{equation}\label{zero-rates}
R_t(\tau) := -\frac{1}{\tau}\left(\phi(\tau,0) + \psi(\tau,0)^\top x_t\right).
\end{equation}
In the market, we can obtain the data for zero rates at  discrete time instants $t_k$ with certain times to maturity $\tau_1,\cdots,\tau_L$, call these data $R_k(\tau_l)$. We believe these data contain noise, hence at time $k$ we observe \[
y_k = [y_k(\tau_1),\cdots,y_k(\tau_L)] = [R_k(\tau_1),\cdots,R_k(\tau_L)]^\top +v_k,
\]
for $k=1,\cdots, K$, and $v_k$ is an $L$-dimensional random vector which presents the noise in the observed data.
Let $0=t_0\leq t_1,\cdots, t_n=T$ be a partition of the time interval $[0,T]$. Then, considering a time-discrete version $x_k := x_{t_k}$, $k\in \mathbb{N}^+$, of the affine process $(x_{t})_{t\geq 0}$, our aim is to derive the parameters of the latent state process $(x_k)_{k\in \mathbb{N}^+}$ given the observations $(y_k)_{k\in \mathbb{N}^+}$.
To be more precise, we consider the following state space model, the observation equation can be seen as of the general form in~\eqref{KF} by enlarging the state vector,
\begin{align}
x_k &= \e^{-A(t_k-t_{k-1})}x_{k-1} + \left(I-\e^{-A(t_k-t_{k-1})}\right)\beta + \int_{t_{k-1}}^{t_k} \e^{-A(t_k-u)}\left(\Sigma+\tilde{\Sigma}\sqrt{x_u^{(1)}}\right)\, \ud W_u\,,\label{X}\\
y_k  &= H_kx_k +H^0_k+ v_k\,\label{Y},
\end{align}
where $I$ is the identity matrix, $(x_k)_{k\in \mathbb{N}^+}$ is the latent process, $(y_k)_{k\in \mathbb{N}^+}$ represents the observations, $H_k$ is a $L\times d$ matrix with each row equal to $-\psi(\tau_l,0)/\tau_l, l=1,\cdots,L$, $H^0_k$ is the column vector $[-\phi(\tau_1,0)/\tau_1,\cdots,-\phi(\tau_L,0)/\tau_L]^\top$ and $v_k$ represents the noise.
The aim is to estimate the model parameters $A, \beta, \Sigma, \tilde{\Sigma}$ given the observation vector $y_{1:k}$ and the variance of $v_k$.
\medskip\\
We end this section by giving some examples of the models of type \eqref{model_analysis} which are well known in the literature and to which we return with numerical experiments in Section~\ref{Numerical results}.
For $d=1$, $p=1$, $\Sigma=0$, $\tilde{\Sigma} \neq 0$, one obtains the {\it Cox-Ingersoll-Ross} (CIR) model, see~\cite{CIR1985}, i.e.,
\begin{align}\label{eq:CIR-model}
\ud x_t =  \alpha_1(\beta -x_t)\, \ud t +\tilde{\Sigma} \sqrt{x_t}\, \ud W_t\,.
\end{align}
For $d=2$, $p=0$, $\beta^{(1)}=\beta^{(2)} =0$, $\Sigma\neq 0$, $\tilde{\Sigma}=0$, one obtains the two-factor {\it Hull-White} model with mean-reversion level $0$, see~\cite{hull1990pricing}, i.e.,
\begin{equation}\label{eq:HW}
\begin{aligned}
\ud x_t^{(1)} &=- \alpha_{11} x^{(1)}_t\, \ud t + \Sigma_{11} \, \ud W_t^{(1)} +\Sigma_{12}\, \ud W_t^{(2)}\,,\\
\ud x_t^{(2)} &=- \alpha_{22} x^{(2)}_t\, \ud t + \Sigma_{21} \, \ud W_t^{(1)} +\Sigma_{22}\, \ud W_t^{(2)}\,.
\end{aligned}
\end{equation}
For $d=2$, $p=q=1$, $\Sigma=0$, one obtains the {\it stochastic volatility} model, see~\cite{heston1993closed}, in which the first component, $(x_t^{(1)})_{t\geq 0}$, represents
the stochastic volatility of the short rate $(x_t^{(2)})_{t\geq 0}$, i.e.,
\begin{equation}\label{eq:HSV}
\begin{aligned}
\ud x_t^{(1)} &=\alpha_{11} (\beta_1-x^{(1)}_t)\, \ud t + \sqrt{x_t^{(1)}}\left(\tilde{\Sigma}_{11} \, \ud W_t^{(1)} +\tilde{\Sigma}_{12}\, \ud W_t^{(2)}\right)\,, \\
\ud x_t^{(2)} &=\alpha_{22} (\beta_2-x^{(2)}_t)\, \ud t + \sqrt{x_t^{(1)}}\left(\tilde{\Sigma}_{21} \, \ud W_t^{(1)} +\tilde{\Sigma}_{22}\, \ud W_t^{(2)}\right)\,.
\end{aligned}
\end{equation}

\section{Kalman particle filter for online parameters estimation}\label{Kalman-particle-filter}

\noindent  In this section, we introduce the {\it Kalman particle filter} for online parameter estimation. It is a semi-recursive algorithm that combines the Kalman filter and the particle filter.
In this new approach, we consider a two layers method as in the RNPF algorithm. In the outer layer, we sample the particles of the model parameters using some Markovian Gaussian kernel which is updated at each time step. In the inner layer, the distribution of the state process and the marginal likelihood $p(y_k\mid y_{1:k-1},\theta_k^{(i)})$, which is used to re-weight the parameter particles in the outer layer, are estimated given the sampled parameter particles.

There are two main differences between our proposed Kalman particle filter algorithm and the RNPF algorithm. The first difference is that in the outer layer we use dynamic jittering functions, i.e.\ the jittering functions can change over time. Specially, in this paper we specify two jittering functions to sample the model parameters, see \eqref{kernel1} and \eqref{kernel2} as described in Subsection~\ref{sec:static-estimation} below. The second difference is that we use the Kalman filter, instead of the particle filter, to update the underlying states in the inner layer.
Note that in  case the state space does not follow  linear Gaussian dynamics, the literature offers different alternatives, see~\cite{bruno2013sequential} for a Monte Carlo approach, or the {\it Gaussian mixture}, see~\cite{sorenson1971recursive}, or Kalman filter extensions such as the {\it extended Kalman filter}, see~\cite{einicke1999robust, wan2001dual}, the {\it unscented Kalman filter}, see~\cite{wan2001unscented} . In these latter methodologies, the idea is to consider an approximation of the state variables which is linear and Gaussian, and then run a Kalman filter on the approximation. When the model is not Gaussian, such an approximation  introduces bias.
In Section~\ref{convergence_analysis}, we will carry a convergence analysis of our algorithm and we will prove that the bias induced by the Gaussian approximation of the model \eqref{X}, \eqref{Y} indeed vanishes when the time step tends to zero.

The use of the two jittering functions in the outer layer and of the Kalman filter in the inner layer allows us to obtain an algorithm that has faster convergence speed and less computational complexity than the RNPF algorithm.
This will be further illustrated in the examples in Section~\ref{Numerical results}.

\subsection{Static model estimation}\label{sec:static-estimation}

\noindent As described in Section~\ref{sec:static-model}, in order to sequentially estimate the posterior density $p(\theta \mid y_{1:k}$), $k=1,\cdots,K$, we face two issues: how to sample particles from the former posterior distribution $p(\theta\mid y_{1:k-1})$ and how to compute the conditional likelihood $p(y_k\mid y_{1:k-1},\theta)$. First, we consider the sampling problem and we introduce the first jittering kernel that we use to update the parameter space $\theta$.

\subsubsection{Gaussian kernel with changing covariance}

Recall that we use the generic notation $N(x;\mu,\Sigma)$ to denote the density at $x$ of the normal distribution with mean vector $\mu$ and covariance matrix $\Sigma$.
We choose a Gaussian kernel such that the conditional density of $\theta_k$ is given as
\begin{align*}
p(\theta_k\mid\theta_{k-1}) = N(\theta_k; \mu(\theta_{k-1}),\Sigma(\theta_{k-1}))\,,
\end{align*}
with $\mu(\theta_{k-1})$ and $\Sigma(\theta_{k-1})$ being respectively the conditional mean and covariance of $\theta_k$ given $\theta_{k-1}$\,. Then if the parameters $\theta_k$ are jittered from this Gaussian kernel, one can easily derive that

\begin{equation}\label{expectation-variance}
\begin{aligned}
\mathbb{E}(\theta_k) &= \mathbb{E}(\mu(\theta_{k-1})), \\
\mbox{Var}(\theta_k)   &= \mathbb{E}(\Sigma(\theta_{k-1}))+\mbox{Var}(\mu(\theta_{k-1}))\,.
\end{aligned}
\end{equation}
Ideally, the jittering should not introduce bias  and information loss (artificial increase in the variance), see~\cite{LW}, which means that $\mathbb{E}(\theta_k)=\mathbb{E}(\theta_{k-1})$ and $\mbox{Var}(\theta_k)=\mbox{Var}(\theta_{k-1})$, $k=1,\cdots,K$. The latter, together with Equations \eqref{expectation-variance} imply
\begin{equation}\label{mu-sigma}
\begin{aligned}
\mathbb{E}(\mu(\theta_{k-1})) &= \mathbb{E}(\theta_{k-1}),\\
\mathbb{E}(\Sigma(\theta_{k-1}))+\mbox{Var}(\mu(\theta_{k-1}))   &= \mbox{Var}(\theta_{k-1})\,.
\end{aligned}
\end{equation}
\noindent To achieve that, the Liu \& West filter \cite{LW} applies a shrinkage to the kernel. We will apply the same technique although the jittering function is used differently in our case. If one assumes a deterministic jittering covariance, i.e.~$\Sigma(\theta_{k-1}) = \Sigma_{k-1}$ and a linear mean function
\begin{equation}\label{eq:a}
\mu(\theta_{k-1}) = a\theta_{k-1} + (1-a)\mathbb{E}(\theta_{k-1}), \mbox{for some $a\in(0,1)$},
\end{equation}
then the jittering kernel satisfying \eqref{mu-sigma} is given by
\begin{equation}\label{kernel1}
p(\theta_k\mid\theta_{k-1}) = N(\theta_k; \mu(\theta_{k-1}), \Sigma(\theta_{k-1}))\,,
\end{equation}
where $ \Sigma(\theta_{k-1})= (1-a^2)\mbox{Var}(\theta_{k-1})$. The kernel in \eqref{kernel1} is the same jittering kernel as used in the Liu \& West filter \cite{LW}. We will refer to the number $a$ as the discount factor.

Before we present our methodology for jittering in detail, we introduce the following assumption which we need in our recursive algorithm later. We will use this assumption in Section~\ref{convergence_analysis} to prove the convergence of our proposed algorithm.

\begin{assumption}\label{ass1}
The jittering kernels $\kappa_k^N(d\theta\mid \theta^\prime)$, for $k=1,\cdots,N$, and $\theta^\prime$ taking values in a compact set $D_\theta \subset \mathbb{R}^d$, satisfy the following inequalities
\begin{equation}\label{A1_1}
\sup_{\theta^\prime\in D_\theta}\int\abs{f(\theta)-f(\theta^\prime)} \kappa_k^N(d\theta\mid \theta^\prime)\leq \frac{e_{1,k}}{\sqrt{N}}\,,
\end{equation}
for some positive constant $e_{1,k}$ and bounded function $f : D_\theta \rightarrow \mathbb{R}$, and
\begin{equation}\label{A1_2}
\sup_{\theta^\prime\in D_\theta}\int\norm{\theta-\theta^\prime}^p \kappa_k^N(d\theta\mid \theta^\prime)\leq \frac{e_{2,k}^p}{\sqrt{N^p}}\,,
\end{equation}
for $p\geq 1$ and some positive constant $e_{2,k}$.
\end{assumption}

\noindent Let $f: D_\theta \rightarrow \mathbb{R}$ be a bounded Lipschitz function. In Proposition~1 of Appendix C in \cite{CrisanMiguez}, set $\epsilon_n=1$ there, it is shown that if for any $p \geq 1$, the jittering kernels $\kappa^N_k(\ud \theta\mid\theta^\prime)$ satisfy
\begin{equation*}
 \sup_{\theta^{\prime} \in D_\theta}\int \norm{\theta-\theta^\prime}^2\kappa^N_k(\ud\theta \mid \theta^{\prime})\leq \frac{c}{\sqrt{N^{p+2}}}\,,
\end{equation*}
for some positive constant $c$ independent of $n$, then Assumption~\ref{ass1} holds.
\medskip\\
The jittering kernels of type \eqref{kernel1} have an appealing property, the covariance can change over time. This aspect helps us to design an algorithm with the following attractive feature. Initially,
since we lack information on the unknown parameters, a larger covariance can lead to a faster convergence of the parameters to the high likelihood area. Over time, the filter refines the estimate of the fixed parameters until at some points a very small variance has been reached which makes the parameter estimation more accurate. However, a direct application of this kernel does not yield a recursive method since  it generally does not satisfy Assumption~\ref{ass1}. Hence, it is unclear whether the algorithm converges. To tackle this issue, we introduce the second Gaussian jittering kernel which satisfies Assumption~\ref{ass1}, hence we can obtain a recursive method if we switch the jittering kernel to the second kernel. The details will be described in the section \ref{Kalman-particle-algorithm}.

\begin{remark}
Although in this paper we have specified the use of Gaussian kernels with a certain mean and variance, the dynamic kernel set up is very generic. In implementation, one can freely choose another kernel that fits its purpose. For example, one can define the variance of the jittering kernel as a monotonically decreasing function of time so that the convergence speed can be manually controlled.
\end{remark}

\subsubsection{Description of the Kalman particle filter methodology}
\label{Kalman-particle-algorithm}

Here  we present our methodology to Gaussian and linear models. When the model is not Gaussian and linear, we approximate it by a Gaussian linear model and hence
we can follow the same methodology as described below to the approximation.\\

\noindent {\bf The non recursive step.}
Assume at time ${k-1}$, one can generate a random grid of samples in the parameter space, say $\{\theta^{(i)}_{k-1}, i=1,\cdots, N\}$.
\begin{itemize}
\item {\bf Jittering step 1.}
Here we apply the kernel (\ref{kernel1}), referred to as jittering kernel 1, to obtain new samples $\{\tilde{\theta}^{(i)}_k, i=1,\cdots, N\}$ (step~\ref{2filteri1} in Algorithm~\ref{algoKPF} below).

\item {\bf Update.} In order to compute the posterior measure $\Gamma_k^{\tilde{\theta}_k^{(i)}}$ at time $k$, one needs to know the mean $B_{k-1}^{\tilde{\theta}^{(i)}_k}$ and the covariance $P_{k-1}^{\tilde{\theta}^{(i)}_k}$ at time ${k-1}$ of the posterior distribution, see Formula~\eqref{eq:posterior-measure} and note that we make the dependence on $\tilde{\theta}^{(i)}_k$ clear in the notation. However, these latter quantities are not available since the parameter has evolved from
$\theta_{k-1}$ at time $k$ to $\tilde{\theta}_{k}$ at time $k$, see also the discussion in Section~\ref{sec:static-model}.
Hence at this step, the algorithm does not run recursively and at every time $k$ where a new parameter particle is sampled, the inner filter re-runs from time $t_0$ to $t_k$ (step~\ref{2filteri2} in Algorithm~\ref{algoKPF}). Moreover, the marginal likelihood $p(y_k\mid y_{1:k-1},\tilde{\theta}_k^{(i)})$, $i=1,\cdots,N$ is computed in the inner filter, using Equation~\eqref{MarginalLikelihood_KF}. This latter will be used to re-weight the parameter particles,   see step~\ref{2filter1iii} in Algorithm~\ref{algoKPF} below.
\item {\bf Resampling.}  We use a resampling technique to obtain $\{\theta^{(i)}_k, B_k^{(i)}, P_k^{(i)}, \,i=1,\cdots,N\}$, further specified in step~\ref{2resample} of Algorithm~\ref{algoKPF}. \\
\end{itemize}
\noindent{\bf The recursive step.}
Once at some time point $t_l$ the variance $\Sigma(\theta_{{l}})=(1-a^2)\mbox{Var}(\theta_{{l}})$ is smaller than a certain level $V_N$ ensuring that Assumption~\ref{ass1} holds (we also set a floor $0<V_f<V_N$ on the jittering variance to prevent the algorithm of getting stuck),
then
\begin{itemize}
\item {\bf Jittering step 2.} We apply the jittering kernel 2
\begin{equation}\label{kernel2}
p(\theta_k^{(i)}\mid\theta^{(i)}_{k-1}) = N(\theta^{(i)}_k;\theta^{(i)}_{k-1},\min\{\max\{\Sigma(\theta_{{l}}),V_f\},V_N\})\,,
\end{equation}
for $k\geq l+1$, see step~\ref{2filterii1} of Algorithm~\ref{algoKPF}.
\item {\bf Update.}
From time $t_{l+1}$ on, we have a recursive algorithm based on the idea  to approximate the posterior measure $\Gamma_k^{\tilde{\theta}_k^{(i)}}$ by $\hat{\Gamma}_k^{{\theta}^{(i)}_{k-1}}$, which is computed using $\tilde{\theta}_{k}^{(i)}$, $B_{k-1}^{\theta^{(i)}_{k-1}}$ and the covariance $P_{k-1}^{\theta^{(i)}_{k-1}}$ (step~\ref{2filteri2} in Algorithm~\ref{algoKPF}). Note that here we use the Kalman filter. Note that, as in Subsection~\ref{RNP-filter}, here we assume that $\Gamma_k^\theta$ is continuous w.r.t.\ $\theta \in D_\theta$. Moreover, the marginal likelihood $p(y_k\mid y_{1:k-1},\tilde{\theta}_k^{(i)})$, $i=1,\cdots,N$ is approximated by the inner filter using Equation (\ref{MarginalLikelihood_KF}), step~\ref{2filter1iii} in Algorithm~\ref{algoKPF}. This latter will be used to re-weight the parameter particles
\item {\bf Resampling.}
Apply a resampling technique to obtain $\{\theta^{(i)}_k, B_k^{(i)}, P_k^{(i)}, \, i=1,\cdots,N\}$, see step~\ref{2resample} of Algorithm~\ref{algoKPF}.
\end{itemize}
We now introduce the Kalman particle algorithm. Recall $\mu(\theta_{k-1})$ and $\Sigma(\theta_{k-1})$ from \eqref{kernel1} and the discount factor $a$ from \eqref{eq:a}.

\noindent
\begin{algorithm}[{\bf Kalman particle filter for static parameter model}]\label{algoKPF}
\mbox{}\medskip\\
{\bf Initialization:}
			\begin{enumerate}
				\item set the number of particles $N$, a value for the discounting factor $a$, a switching variance level $V_N$ and a floored variance level $V_f$,
				\item assume an initial distribution $p(\theta_{0})$ for the parameters,
				\item sample from the initial distribution to get $N$ particles $\{\theta_{0}^{(i)},\, i=1,\cdots,N\}$ for the parameters,
				\item for each particle $\theta_{0}^{(i)}$, assign the same initial mean $B_{0}^{(i)}$ and covariance value $P_{0}^{(i)}$ of the posterior distribution
				$\Gamma_{0}^{\theta_{0}^{(i)}}$,	$i=1,\cdots, N$.			
			\end{enumerate}			
\medskip
{\bf Recursion:}\begin{enumerate}
	\item Filtering: given $\{\theta_{k-1}^{(i)},\, i=1,\cdots,N\}$,
				\begin{enumerate}
	      	\item \label{init} if $\Sigma(\theta_{k-1})>V_N$ (jittering case 1), for each $i=1,\cdots,N$,

				\begin{enumerate}
					\item\label{2filteri1} sample new parameters from the kernel (\ref{kernel1}), i.e.
							\begin{equation*}
\tilde{\theta}_k^{(i)}\sim N(\mu(\theta_{k-1}^{(i)}),\Sigma(\theta_{k-1}))\,,
							\end{equation*}
					\item\label{2filteri2} based on the parameter $\tilde{\theta}_k^{(i)}$, use the Kalman filter to compute the mean and covariance of the posterior distribution
					from time $1$ to ${k}$ and hence obtain $\Gamma_k^{\tilde{\theta}^{(i)}_k}$,
				\end{enumerate}

					\item\label{once} once $\Sigma(\theta_{l-1})\leq V_N$ (jittering case 2), for some $t_l$, then for $k\geq l+1$ and $i=1,\cdots,N$, given $\{B_{k-1}^{(i)}, P_{k-1}^{(i)}\}$,

				\begin{enumerate}
					\item\label{2filterii1} sample new parameters from the kernel (\ref{kernel2}), i.e.
							\begin{equation*}
									\tilde{\theta}_k^{(i)}\sim N(\theta_{k-1}^{(i)},\min\{\max\{\Sigma(\theta_{{k-1}}),V_f\},V_N\})\,,
							\end{equation*}
					\item based on the parameters $\tilde{\theta}_{k}^{(i)}$, $B_{k-1}^{(i)}$ and $P_{k-1}^{(i)}$, use the Kalman filter to compute the mean $\tilde{B}_{k-1}^{(i)}$ and covariance $\tilde{P}_{k-1}^{(i)}$ of the posterior distribution at time $k$ and hence obtain an approximation $\hat{\Gamma}_k^{{\theta}^{(i)}_{k-1}}$ of the posterior distribution (update step),
				\end{enumerate}	
				\item\label{2filter1iii} if \ref{init} or \ref{once} holds, then for $i=1,\cdots, N$, compute  $\tilde{p}(y_{k}\mid y_{1:k-1},\tilde{\theta}_k^{(i)})$ using Equation (\ref{MarginalLikelihood_KF}), consequently, using \eqref{theta-weights}), obtain an approximation of the normalized weights given by
				$$\tilde{w}_k^{\tilde{\theta}_k^{(i)}} = \frac{\tilde{p}(y_{k}\mid y_{1:k-1},\tilde{\theta}_k^{(i)})}{\sum_{i=1}^{N}\tilde{p}(y_{k}\mid y_{1:k-1},\tilde{\theta}_k^{(i)})}\,,$$
				which gives un update to be used in the next step,
				
				\end{enumerate}
	\item\label{2resample} Resampling: for each $i=1,\cdots, N$, set $\{\theta_k^{(i)}, B_k^{(i)}, P_k^{(i)}\}$   equal to $\{\tilde{\theta}_k^{(p)}, \tilde{B}_k^{(p)}, \tilde{P}_k^{(p)}\}$, with probability $\tilde{w}_k^{\tilde{\theta}_k^{(p)}}$, where $p \in \{1,\cdots,N\}$.
	\item Return to the filtering step.		
\end{enumerate}
\end{algorithm}
\noindent
Note that Algorithm~\ref{algoKPF} could be applied to general state space models \eqref{state_space_model}. For our convergence analysis and in the financial applications,
we will focus on the special type \eqref{X}, \eqref{Y} of affine state space models.
When $\Sigma \neq 0$, the transition measure $\pi^\theta_k(\ud x\mid x_{k-1})$ of \eqref{X} is not Gaussian, we need to approximate it by a Gaussian transition in order to apply the Kalman filter. The approximation is obtained by replacing $\sqrt{x_u^{(1)}}$ in \eqref{X} with $\sqrt{x^{(1)}_{t_{k-1}}}$, resulting in
\begin{align}\label{Xhat}
\check{x}_k &= \e^{-A(t_k-t_{k-1})}x_{k-1} + \left(I-\e^{-A(t_k-t_{k-1})}\right)\beta + \int_{t_{k-1}}^{t_k} \e^{-A(t_k-u)}\left(\Sigma+\tilde{\Sigma}\sqrt{x_{k-1}^{(1)}}\right)\,\ud W_u.
\end{align}
Note that given $x_{k-1}$, the variable $\check{x}_k$
admits a Gaussian transition for $k \in \mathbb{N}^+$. Hence, given the model parameters $\theta$ (i.e.~$A,\beta, \Sigma_1, \Sigma_2$), we can compute
the approximated transition measure $\hat{\pi}_k^{\theta}=p(d\check{x}_k\mid x_{k-1}, \theta)$.
Recall from Equations \eqref{eq:posterior-theta}, (\ref{int_gamma}) and (\ref{MarginalLikelihood}) that the weights for the parameters space are computed by
\begin{equation}\label{eq:weights}
w_k^{\tilde{\theta}_k^{(i)}}\propto \left((l^{\tilde{\theta}_k^{(i)}}_{y_k},\pi_k^{\tilde{\theta}_k^{(i)}}),\Gamma_{k-1}^{\tilde{\theta}_k^{(i)}}\right)\,.
\end{equation}
In the recursive step of Algorithm~\ref{algoKPF}, the measure $\Gamma_{k-1}^{\tilde{\theta}_k^{(i)}}$ is not available at time ${k-1}$, since the parameter evolves from $\theta_{k-1}^{(i)}$ to $\tilde{\theta}_k^{(i)}$ in the jittering step at time $k$. In order to have a recursive algorithm, the estimate $\hat{\Gamma}_{k-1}^{{\theta}_{k-1}^{(i)}}$ obtained at time ${k-1}$ is used to approximate the measure $\Gamma_{k-1}^{\tilde{\theta}_k^{(i)}}$. Hence, when the model is linear and Gaussian, we obtain the following estimation of the weights for the parameters space
\begin{equation}\label{weights-for-theta1}
\tilde{w}_k^{\tilde{\theta}_k^{(i)}}\propto \left((l^{\tilde{\theta}_k^{(i)}}_{y_k},\pi_k^{\tilde{\theta}_k^{(i)}}),\hat{\Gamma}_{k-1}^{\theta_{k-1}^{(i)}}\right)\,.
\end{equation}
When the model is nonlinear or non-Gaussian, then we consider the approximation \eqref{Xhat} to \eqref{X}. In this case, the transition probability $\pi_k^{\tilde{\theta}_k^{(i)}}$ is approximated by a Gaussian transition probability $\hat{\pi}_k^{\tilde{\theta}_k^{(i)}}$. Therefore, in the non-recursive step, estimation of the weights for the parameters space is given by
\begin{equation}\label{weights-for-theta2}
\hat{w}_k^{\tilde{\theta}_k^{(i)}}\propto \left((l^{\tilde{\theta}_k^{(i)}}_{y_k},\hat{\pi}_k^{\tilde{\theta}_k^{(i)}}),\hat{\Gamma}_{k-1}^{\tilde{\theta}_{k}^{(i)}}\right)\,,
\end{equation}
In the recursive step, the measure $\Gamma_{k-1}^{\tilde{\theta}_k^{(i)}}$ is approximated by $\hat{\Gamma}_{k-1}^{{\theta}_{k-1}^{(i)}}$. Hence, we obtain the following estimation of the weights for the parameters space
\begin{equation}\label{weights-for-theta}
\hat{w}_k^{\tilde{\theta}_k^{(i)}}\propto \left((l^{\tilde{\theta}_k^{(i)}}_{y_k},\hat{\pi}_k^{\tilde{\theta}_k^{(i)}}),\hat{\Gamma}_{k-1}^{{\theta}_{k-1}^{(i)}}\right)\,,
\end{equation}
together with the estimation of the posterior measure of $x_k$, for $A \in \mathcal{B}(\mathbb{R}^d)$ a Borel set,
\begin{align*}
\hat{\Gamma}_{k}^{\tilde{{\theta}}_{k}^{(i)}} (A)= \frac{\left((\mathbf{1}_{\{x_k \in A\}}l^{\tilde{\theta}_k^{(i)}}_{y_k},\hat{\pi}_k^{\tilde{\theta}_k^{(i)}}),\hat{\Gamma}_{k-1}^{{\theta}_{k-1}^{(i)}}\right)}{\left((l^{\tilde{\theta}_k^{(i)}}_{y_k},\hat{\pi}_k^{\tilde{\theta}_k^{(i)}}),\hat{\Gamma}_{k-1}^{{\theta}_{k-1}^{(i)}}\right)}\,.
\end{align*}
We will make use of the weights  \eqref{weights-for-theta1}, \eqref{weights-for-theta2} and \eqref{weights-for-theta} in our convergence analysis later in Section~\ref{convergence_analysis}.

\subsection{Kalman Particle filter for models with piece-wise constant parameters}

So far in this paper, the model parameters are assumed to be fixed over time. But in many applications, it is more realistic to assume that, at least, some parameters are time-varying, for instance if they are piecewise constant. Offline estimation methods such as MLE are able to deal with this situation only if the change point of the parameter is known beforehand. This does not hold in most of the cases. For the particle filter methods which treat the model parameters as static, the variance of the samples for the parameters decreases with more observed data. Hence the marginal distribution of the model parameters will be increasingly concentrated around certain values. The consequence is that the particle filter algorithm is not able to capture abrupt changes of parameters.

We extend our proposed algorithm for static parameter to adapt to abrupt changes of parameters. To achieve that, we first identify the change points. This step is done by comparing the marginal likelihood between two consecutive steps. Suppose the parameter samples have already converged to the actual value. If at some point the actual parameter value jumps to another value, then the marginal likelihood based on the existing parameter samples are far from optimal. Hence the marginal likelihood at this time point should be significantly smaller than that at the previous time point. On the other hand, if at this time point the actual parameter value does not change, then the marginal likelihood should also be very close to the previous value. So we set a threshold $b<1$ and if at some point time ${k}$, the maximum marginal likelihood for $\theta^{(i)}_k \in D_\theta$, $k\in \mathbb{N}^+$, satisfies
$$\max_{1\leq i\leq  N}\left\{p(y_{k}\mid y_{1:k-1},\theta^{(i)}_{k})\right\}<b\max_{1\leq i\leq N}\left\{p(y_{k-1}\mid y_{1:k-2},\theta^{(i)}_{k-1})\right\}\,,$$ then we consider the time point $k$ to be the change point of the parameters. Of course, this condition is not sufficient but it is necessary.

Another issue here is that the parameter samples may have already (nearly) converged before the change point, hence the variance of these samples is too small to capture the change. This problem can be tackled by adding new samples from the initial parameter distribution to increase the sample variance. But since the variance is increased, the jittering kernel (\ref{kernel2}) does not satisfy Assumption~\ref{ass1}. Hence the jittering kernel should switch to (\ref{kernel1}). Moreover, since the model parameter changes, the posterior distributions from the previous time point are also not valid anymore, and a new initial value for the mean and the variance of the posterior distribution should also be initialized. So once the change point is determined, one can treat the calibration of the model as a new calibration based on data after this change point.
\medskip\\
We introduce the Kalman particle algorithm extended to time-varying parameters. We present the algorithm in full detail, noting that the differences with the previous Algorithm~\ref{algoKPF} are in the two jittering cases in the filtering step.

\noindent
\begin{algorithm}[{\bf Kalman particle filter for models with time-varying parameters}]\label{algoKPFJ}
\mbox{}\medskip\\
{\bf Initialization:}
			\begin{enumerate}
				\item set the number of particles $N$, a value for the discounting factor $a$, a switching variance level $V_N$, a floored variance level $V_f$, and the threshold parameter $b$,
				\item assume an initial distribution $p(\theta_{0})$ for the parameters,
				\item sample from the initial distribution to get $N$ particles $\{\theta_{0}^{(i)}, i=1,\cdots,N\}$ for the parameters,
				\item for each particle $\{\theta_{0}^{(i)}\}$, assign the same initial mean and covariance value of the posterior distribution $B_{0}^{(i)}$ and $P_{0}^{(i)}$ for the Kalman filter update,
			\end{enumerate}			
\medskip
{\bf Recursion:}
\begin{enumerate}
\item Filtering: given $\{\theta_{k-1}^{(i)}, \, i=1,\cdots,N\}$,
	      \begin{enumerate}
	      	\item if $\Sigma(\theta_{k-1})>V_N$ (jittering case 1), for each $i=1,\cdots,N$,
	\begin{enumerate}
					\item sample new parameters from the kernel (\ref{kernel1}), i.e.,
							\begin{equation*}
									\tilde{\theta}_k^{(i)}\sim N(\tilde{\theta}_k^{(i)}; \mu(\theta_{k-1}^{(i)}), \Sigma(\theta_{k-1}))\,,
							\end{equation*}
					\item based on the parameter $\tilde{\theta}_k^{(i)}$, use the Kalman filter to compute the mean and the covariance of the posterior distribution from time $1$ to ${k}$,

                    \item compute the likelihood $p(y_{k}\mid y_{1:k-1},\tilde{\theta}_k^{(i)})$, consequently obtain the normalized weights
                    $$w_k^{\tilde{\theta}^{(i)}_k}  = \frac{p(y_{k}\mid y_{1:k-1},\tilde{\theta}_k^{(i)})}{\sum_{i=1}^Np(y_{k}\mid y_{1:k-1},\tilde{\theta}_k^{(i)})}\,;$$
				\end{enumerate}

		\item once $\Sigma(\theta_{k-1})\leq V_N$ for some $t_l$ (jittering case 2), then for $k\geq l+1$ and $i=1,\cdots, N$, given $\{B_{k-1}^{(i)}, P_{k-1}^{(i)}\}$,

				\begin{enumerate}
					\item sample new parameters from the kernel (\ref{kernel2}), i.e.
							\begin{equation*}
									\tilde{\theta}_k^{(i)}\sim N(\tilde{\theta}_k^{(i)}; \theta_{k-1}^{(i)},\min\{\max\{\Sigma(\theta_{k-1}),V_f\},V_N\}),
							\end{equation*}
					\item based on the parameters $\tilde{\theta}_{k}^{(i)},B_{k-1}^{(i)}$ and $P_{k-1}^{(i)}$, use the  Kalman filter to compute the mean $\tilde{B}_k^{(i)}$ and the covariance $\tilde{P}_k^{(i)}$ of the approximated posterior distribution at time $k$,
                    \item compute an approximation $\tilde{p}(y_{k}\mid y_{1:k-1},\tilde{\theta}_k^{(i)})$ of the likelihood $p(y_{k}\mid y_{1:k-1},\tilde{\theta}_k^{(i)})$. If
                    $$\max_{1\leq i\leq N}\{\tilde{p}(y_{k}\mid y_{1:k-1},\tilde{\theta}_{k}^{(i)})\}<b\max_{1\leq i\leq N}\{\tilde{p}(y_{k-1}\mid y_{1:k-2},\tilde{\theta}_{k-1}^{{(i)}})\}\,,$$ then go to  the Initialization step of the algorithm to initialize the algorithm using the data after time $k$. Otherwise compute the normalized weights
                    				$$\tilde{w}_k^{\tilde{\theta}_k^{(i)}} = \frac{\tilde{p}(y_{k}\mid y_{1:k-1},\tilde{\theta}_k^{(i)})}{\sum_{i=1}^{N}\tilde{p}(y_{k}\mid y_{1:k-1},\tilde{\theta}_k^{(i)})}\, .$$
				\end{enumerate}	
	\end{enumerate}
				\item Resampling: for each $i=1,\cdots, N$, set $\{{\theta}_k^{(i)}, B_k^{(i)}, P_k^{(i)}\} =\{\tilde{\theta}_k^{(p)}, \tilde{B}_k^{(p)}, \tilde{P}_k^{(p)}\}$, with probability
				$w_k^{\tilde{\theta}^{(p)}_k}$, where $p \in \{1,\cdots, N\}$.

\end{enumerate}
\end{algorithm}
\noindent
As in this setting the parameters are time-varying, there is of course no point in providing a convergence analysis. But in Section~\ref{Numerical results} we will present a numerical study where Algorithm~\ref{algoKPFJ} is seen to be capable of tracking a sudden parameter change.

\section{Convergence analysis}\label{convergence_analysis}

In the Kalman particle Algorithm~\ref{algoKPF} introduced in Section~\ref{Kalman-particle-algorithm}, the (conditional) measure
\begin{equation}\label{original-measure}
\mu_k(\ud\theta_k) = p(\ud\theta_k\mid y_{1:k})\,,
\end{equation}
 is estimated for $k=1,\cdots,K$.
From now on, for convenience, we assume the algorithm starts at time $1$. At each time, the algorithm has three main steps: jittering, update and resampling. \\
Define the maximum step size $\delta = \sup_{k=1,\cdots, K}(t_{k}-t_{k-1})$. Let the parameter $\theta \in D_\theta$.  At time $k$, suppose the estimated measure of the last time $\mu_{k-1}^N$ is available and
\begin{equation*}
\mu_{k-1}^N = \frac{1}{N}\sum_{i=1}^N\delta_{\theta_{k-1}^{(i)}}(\ud\theta)\,.
\end{equation*}
In the jittering step as described in Algorithm~\ref{algoKPF}, new samples $\tilde{\theta}_{k}^{(i)}$ are sampled from the kernel function $N(\theta_{k-1}^{(i)},\Sigma(\theta_{t_{l}}))$. The resulting measure $\tilde{\mu}_k^N$ is then defined by
\begin{equation}\label{eq:mu-tilde}
\tilde{\mu}_k^N = \frac{1}{N}\sum_{i=1}^N\delta_{\tilde{\theta}_{k}^{(i)}}(\ud\theta)\,.
\end{equation}
In this step, no extra information is used to refine the estimates on the parameters. Hence the aim is to prove that the measure $
\tilde{\mu}_k^N$ converges to the measure $\mu_{k-1}^N$ in some sense when $\delta$ goes to zero and the number of samples $N$ goes to infinity.

In the update step as described in Algorithm~\ref{algoKPF} of Section~\ref{Kalman-particle-algorithm}, there are four cases to analyze, combinations of Gaussian-linear or non-Gaussian/non-linear models and recursive or non-recursive parts of the algorithm. When the model is linear and Gaussian, and we consider the non-recursive step of the algorithm, then  the normalized weights $w_k^{\tilde{\theta}_k^{(i)}}$ are computed exactly. When we consider the recursive step of the algorithm for a Gaussian linear model, the normalized weights $w_k^{\tilde{\theta}_k^{(i)}}$ are estimated by $\tilde{w}_k^{\tilde{\theta}_k^{(i)}}$, see \eqref{weights-for-theta1}. For the non-Gaussian and non-linear model, we consider the approximation \eqref{Xhat} to \eqref{X}, the normalized weights are estimated by \eqref{weights-for-theta2} and \eqref{weights-for-theta} in the non-recursive step and recursive step, respectively . Since the convergence of the weights in this latter case implies the convergence of the weights in the first three cases, we will only consider the model \eqref{Xhat} and the recursive step of the algorithm. In this case, we obtain a new estimated measure
\begin{equation}\label{eq:mu-hat}
\hat{\mu}_k^N = \sum_{i=1}^N \hat{w}_k^{\tilde{\theta}_k^{(i)}}\delta_{\tilde{\theta}_{k}^{i}}(\ud\theta)\,
\end{equation}
\noindent and the convergence analysis of $\hat{\mu}_k^N$ depends on that of the estimated weights $\hat{w}_k^{\tilde{\theta}_k^{(i)}}$.
The aim is to prove that $\hat{\mu}_k^N$ converges to $\mu_k$ in some sense when $\delta$ goes to $0$ and $N$ goes to infinity.

In the resampling step, the new particles $\{\theta_k^{(i)},\, i=1,\cdots,N\}$ are sampled from the empirical distribution $
\sum_{i=1}^{N}\hat{w}_k^{\tilde{\theta}_k^{(i)}}\delta_{\tilde{\theta}_k^{(i)}}$ and we need to prove that the measure
\begin{equation}\label{approximating-measure}
\mu_k^N = \frac{1}{N}\sum_{i=1}^N\delta_{\theta_k^{(i)}}(\ud \theta)\,,
\end{equation}
converges to $\mu_k$ when $\delta$ goes to zero and $N$ goes to infinity.
\medskip\\
For our convergence analysis we will need besides   Assumption~\ref{ass1}, the following assumptions.

\begin{assumption}\label{ass2}
The posterior measure $\Gamma_k^{\theta}$, $k=1,\cdots,N$, $\theta\in D_\theta$, is Lipschitz in the parameter $\theta$, i.e.~for any bounded continuous function $f$,
\begin{equation}\label{A2}
\abs{(f,\Gamma_k^{\theta})-(f,\Gamma_k^{\theta^\prime})}\leq e_{3,k}\norm{\theta-\theta^\prime}\,,
\end{equation}
for some constant $e_{3,k}$ .
\end{assumption}

\begin{assumption}\label{ass3}
Let
$(x_k)_{k\in \mathbb{N}^+}$ be as in \eqref{X}. There exist a constant $M>0$ such that
$$\sup_{\theta\in D_\theta, \, k\in \mathbb{N}^+}\mathbb{E}\left[\abs{x_k}\mid \theta\right]\leq M.$$
\end{assumption}

\begin{assumption}\label{ass4}
For any fixed observation sequence $y_{1:k}$, the likelihood $\{l^{\theta}_{y_t}(x),\,t=1,\cdots,k,\,\theta\in D_\theta\}$ satisfies
\begin{enumerate}
  \item $l^{\theta}_{y_t}(x)$ is continuous w.r.t.\ $x$ for all $\theta\in D_\theta$,
  \item $\norm{l_{y_t}}_\infty := \sup_{\theta\in D_\theta}\norm{l^{\theta}_{y_t}}_\infty<\infty$,
  \item $\inf_{\theta\in D_\theta}l^{\theta}_{y_t}>0$.
\end{enumerate}
\end{assumption}

\noindent The Inequality (\ref{A1_1}) in Assumption~\ref{ass1} is used for the convergence analysis of the jittering step, the Inequality (\ref{A1_2}) in Assumption~\ref{ass1}, together with Assumptions~\ref{ass2} and~\ref{ass3} are specially used for the analysis for the update step and Assumption~\ref{ass4} is used for the analysis for both update and resampling steps. Given the convergence of $\mu_{k-1}^N$ to $\mu_{k-1}$ in an $L^p$-sense when $\delta$ goes to zero and $N$ goes to infinity, we present the convergence of the measures $\tilde{\mu}_k^N$, $\hat{\mu}_k^N$ and $\mu_k^N$ respectively in the three Lemmas~\ref{jittering}, \ref{update}, \ref{resampling} in the subsections below. In Section~\ref{section:convergence}, we study the convergence of Algorithm~\ref{algoKPF} presented in Subsection~\ref{Kalman-particle-algorithm} by induction using these three lemmas. The main result there is Theorem~\ref{theorem:convergence}. Our analysis in inspired by the results in \cite{CrisanMiguez} for the nested particle filter and follows a similar pattern. Note however the crucial differences between our Algorithm~\ref{algoKPF} and their nested particle filter. We have to deal with  only one layer with a particle filter (instead of two such layers), as we use the Kalman filter in the other layer, but we also pay attention to the time discretization of the continuous processes.

\subsection{Jittering}

In the jittering step~\ref{2filterii1} for low values of $\Sigma(\theta_{k-1})$ the new parameter particles $\tilde{\theta}_k^{(i)}$ are sampled from a kernel function $\kappa_k^N(\ud\theta\mid\theta_{k-1}^{(i)})$, $i=1,\cdots,N$.
The following lemma shows that the error due to the jittering step vanishes. This lemma can be seen as the analog of Lemma~3 in \cite{CrisanMiguez}, but with the term $\tfrac{1}{\sqrt{M}}$ there (which plays no role in our analysis) replaced with $\sqrt{\delta}$ as we treat  the influence of time discretization. It can be proven in the same way and we present it here for the sake of completeness  and in a form that suits our purposes.

\begin{lem}\label{jittering}
Let $f$ be a bounded function and suppose Assumption~\ref{ass1} holds. If
\begin{equation}\label{pre_mu}
\norm{(f,\mu_{k-1}^N)-(f,\mu_{k-1})}_p\leq \frac{c_{1,{k-1}}}{\sqrt{N}}+d_{1,{k-1}}\sqrt{\delta}\,,
\end{equation}
for some constants $c_{1,{k-1}}$ and $d_{1,{k-1}}$ which are independent of $N$ and $\delta$, then there exist constants $\tilde{c}_{1,{k}}$ and $\tilde{d}_{1,t}$ which are independent of $N$ and $\delta$, such that
\begin{equation}\label{mu_tilde}
\norm{(f,\tilde{\mu}_{k}^N)-(f,\mu_{k-1})}_p\leq \frac{\tilde{c}_{1,{k}}}{\sqrt{N}}+\tilde{d}_{1,{k}}\sqrt{\delta}\,.
\end{equation}
\end{lem}

\subsection{Update}

To prove the convergence in the update step~\ref{2filter1iii}, we first need to prove that the error introduced from the approximation of the weights $w_k^{\tilde{\theta}_{k}^{(i)}}$ by the weights $\hat{w}_k^{\tilde{\theta}_{k}^{(i)}}$, can be bounded by a desired quantity as in (\ref{mu_tilde}). The following lemma is a core result in our convergence analysis, in the proof of it we exploit the affine nature of the state process.

\begin{lem}\label{lemma2}
Let the observation sequence $y_{1:k}$ be fixed. Suppose function $f$ is bounded and continuous and Assumptions~\ref{ass1},\ref{ass2}, and~\ref{ass3} hold. If
\begin{equation}\label{pre_Gamma}
\sup_{1\leq i\leq N} \abs{(f,\hat{\Gamma}_{k-1}^{\theta_{k-1}^{(i)}})-(f,{\Gamma}_{k-1}^{\theta_{k-1}^{(i)}})}\leq \frac{c_{2,{k-1}}}{\sqrt{N}}+d_{2,{k-1}}\sqrt{\delta}\,,
\end{equation}
for some constants $c_{2,{k-1}}$ and $d_{2,{k-1}}$ which are independent of $N$ and $\delta$, then there exist constants $\tilde{c}_{2,{k}}$ and $\tilde{d}_{2,k}$ which are independent of $N$ and $\delta$ such that
\begin{equation}\label{gamma}
\sup_{1\leq i\leq N} \abs{((f,\hat{\pi}_k^{\tilde{\theta}_{k}^{(i)}}), \hat{\Gamma}_{k-1}^{\theta_{k-1}^{(i)}})
-((f,\pi_k^{\tilde{\theta}_{k}^{(i)}}),{\Gamma}_{k-1}^{\tilde{\theta}_{k}^{(i)}})}\leq \frac{\tilde{c}_{2,{k}}}{\sqrt{N}}+\tilde{d}_{2,{k}}\sqrt{\delta}\,.\\
\end{equation}
\end{lem}
\noindent
To prove Lemma~\ref{lemma2}, we exploit the structure (\ref{X}) of the state process. Besides, we need the following auxiliary result. The proof of it follows the same lines as the proof of Lemma~4 in \cite{CrisanMiguez}, but note again the $\frac{1}{\sqrt{M}}$ term is replaced by $\sqrt{\delta}$.
\begin{lem}
Suppose the function $f$ is bounded and continuous. Moreover, suppose Assumptions~\ref{ass1} and \ref{ass2} and Inequality \eqref{pre_Gamma} hold.
 Then there exist some constants $\tilde{c}_{2,{k-1}}$ and $\tilde{d}_{2,{k-1}}$ which are independent of $N$, $\delta$ and of all $\theta$ such that
\begin{equation}\label{tilde_Gamma}
\sup_{1\leq i\leq N} \abs{(f,\hat{\Gamma}_{k-1}^{\theta_{k-1}^{(i)}})-(f,{\Gamma}_{k-1}^{\tilde{\theta}_{k}^{(i)}})}\leq \frac{\tilde{c}_{2,{k-1}}}{\sqrt{N}}+\tilde{d}_{2,{k}}\sqrt{\delta}\,.
\end{equation}
\end{lem}
\noindent
Now we are ready to prove Lemma~\ref{lemma2}.
\begin{proof}
Using the triangle inequality, one obtains

\begin{align}\label{lemma2_term12}
\lefteqn{\sup_{1\leq i\leq N}\abs{((f,\hat{\pi}_k^{\tilde{\theta}_{k}^{(i)}}), \hat{\Gamma}_{k-1}^{\theta_{k-1}^{(i)}})
-((f,\pi_k^{\tilde{\theta}_{k}^{(i)}}),{\Gamma}_{k-1}^{\tilde{\theta}_{k}^{(i)}})}\leq} \nonumber\\
& \qquad\qquad
\sup_{1\leq i\leq N}\abs{((f,\hat{\pi}_k^{\tilde{\theta}_{k}^{(i)}}), \hat{\Gamma}_{k-1}^{\theta_{k-1}^{(i)}})
-((f,\hat{\pi}_k^{\tilde{\theta}_{k}^{(i)}}), {\Gamma}_{k-1}^{\tilde{\theta}_{k}^{(i)}})}
\nonumber\\
& \qquad\qquad
+\sup_{1\leq i\leq N}\abs{((f,\hat{\pi}_k^{\tilde{\theta}_{k}^{(i)}}), {\Gamma}_{k-1}^{\tilde{\theta}_{k}^{(i)}})
-((f,\pi_k^{\tilde{\theta}_{k}^{(i)}}),{\Gamma}_{k-1}^{\tilde{\theta}_{k}^{(i)}})}\,.
\end{align}
Note that $(f,\hat{\pi}_k^{\tilde{\theta}_{k}^{(i)}})$ is bounded by $\norm{f}_\infty$. Hence, using Inequality (\ref{tilde_Gamma}), we get for the first term on the right hand side of (\ref{lemma2_term12})
\begin{equation}\label{lemma2_fisrt}
\sup_{1\leq i\leq N}\abs{((f,\hat{\pi}_k^{\tilde{\theta}_{k}^{(i)}}), \hat{\Gamma}_{k-1}^{\theta_{k-1}^{(i)}})
-((f,\hat{\pi}_k^{\tilde{\theta}_{k}^{(i)}}), {\Gamma}_{k-1}^{\tilde{\theta}_{k}^{(i)}})}\leq \frac{\tilde{c}_{2,{k-1}}}{\sqrt{N}}+\tilde{d}_{2,{k}}\sqrt{\delta}\,.
\end{equation}
Recall that $(x_k)_{k\in \mathbb{N}}$ and $(\check{x}_k)_{k\in \mathbb{N}}$ respectively from \eqref{X} and \eqref{Xhat}. For  $\delta$, $M_1>0$, define the sets

\begin{equation*}
\begin{aligned}
A_\delta &= \left\{\abs{\check{x}_k^{\tilde{\theta}_{k}^{(i)}}-x_k^{\tilde{\theta}_{k}^{(i)}}}<\delta,\, 1\leq i\leq N \right\}\,, \\
B_{M_1}  &= \left\{\abs{\check{x}_k^{\tilde{\theta}_{k}^{(i)}}}\leq M_1, \,\,\abs{{x}_k^{\tilde{\theta}_{k}^{(i)}}}\leq M_1,\, 1\leq i\leq N\right\}\,.
\end{aligned}
\end{equation*}
Note that $(f,\hat{\pi}_k^{\tilde{\theta}_{k}^{(i)}})$ can be seen as the (conditional) expectation of $f(.)$ taken under the measure $\hat{\pi}_k^{\tilde{\theta}_{k}^{(i)}}$. Stated otherwise, we can see it as the expectation of $f(\check{x}_k^{\tilde{\theta}_{k}^{(i)}})$. Likewise, we can see $(f,\pi_k^{\tilde{\theta}_{k}^{(i)}})$ as the (conditional) expectation of $f({x}_k^{\tilde{\theta}_{k}^{(i)}})$. Below we use the notations $\mathbb{E}f(\check{x}_k^{\tilde{\theta}_{k}^{(i)}})$ and $\mathbb{E}f({x}_k^{\tilde{\theta}_{k}^{(i)}})$ for these expectations.
With these interpretations, the second term on the right hand side of (\ref{lemma2_term12}) yields
\begin{align}
&\sup_{1\leq i\leq N}\abs{\left((f,\hat{\pi}_k^{\tilde{\theta}_{k}^{(i)}}), {\Gamma}_{k-1}^{\tilde{\theta}_{k}^{(i)}}\right)
-\left((f,\pi_k^{\tilde{\theta}_{k}^{(i)}}),{\Gamma}_{k-1}^{\tilde{\theta}_{k}^{(i)}}\right)}\nonumber\\
 &\qquad =
\sup_{1\leq i\leq N}\abs{\left((f,\hat{\pi}_k^{\tilde{\theta}_{k}^{(i)}})-(f,\pi_k^{\tilde{\theta}_{k}^{(i)}}), {\Gamma}_{k-1}^{\tilde{\theta}_{k}^{(i)}}\right)}\nonumber\\
&\qquad \leq \sup_{1\leq i\leq N}{\left(\abs{(f,\hat{\pi}_k^{\tilde{\theta}_{k}^{(i)}})-(f,\pi_k^{\tilde{\theta}_{k}^{(i)}})}, {\Gamma}_{k-1}^{\tilde{\theta}_{k}^{(i)}}\right)}\nonumber\\
&\qquad = \sup_{1\leq i\leq N}{\left(\abs{\mathbb{E}f(\check{x}_k^{\tilde{\theta}_{k}^{(i)}})-\mathbb{E}f({x}_k^{\tilde{\theta}_{k}^{(i)}})}, {\Gamma}_{k-1}^{\tilde{\theta}_{k}^{(i)}}\right)}\nonumber\\
&\qquad \leq \sup_{1\leq i\leq N}\left(\mathbb{E}\abs{f(\check{x}_k^{\tilde{\theta}_{k}^{(i)}})-f({x}_k^{\tilde{\theta}_{k}^{(i)}})}, {\Gamma}_{k-1}^{\tilde{\theta}_{k}^{(i)}}\right)\nonumber\\
&\qquad = \sup_{1\leq i\leq N}{\left(\mathbb{E}\mathbf{1}_{A_\delta}\mathbf{1}_{B_{M_1}^c}\abs{f(\check{x}_k^{\tilde{\theta}_{k}^{(i)}})-f({x}_k^{\tilde{\theta}_{k}^{(i)}})}, {\Gamma}_{k-1}^{\tilde{\theta}_{k}^{(i)}}\right)}\nonumber\\
&\qquad \qquad  + \sup_{1\leq i\leq N}{\left(\mathbb{E}\mathbf{1}_{A_\delta}\mathbf{1}_{B_{M_1}}\abs{f(\check{x}_k^{\tilde{\theta}_{k}^{(i)}})-f({x}_k^{\tilde{\theta}_{k}^{(i)}})}, {\Gamma}_{k-1}^{\tilde{\theta}_{k}^{(i)}}\right)}\nonumber\\
&\qquad \qquad  + \sup_{1\leq i\leq N}{\left(\mathbb{E}\mathbf{1}_{A_\delta^c}\abs{f(\check{x}_k^{\tilde{\theta}_{k}^{(i)}})-f({x}_k^{\tilde{\theta}_{k}^{(i)}})}, {\Gamma}_{k-1}^{\tilde{\theta}_{k}^{(i)}}\right)}.\label{lemma2_term2}
\end{align}
We need to find an upper bound for the three terms on the right hand side of (\ref{lemma2_term2}).

Consider the first term and define $M_1=M/\sqrt{\delta}$. Using Assumption~\ref{ass3} and the Markov inequality, we get
\begin{equation}\label{lemma2_21}
\begin{aligned}
\sup_{1\leq i\leq N}\mathbb{E}\mathbf{1}_{A_\delta}\mathbf{1}_{B_{M_1}^c}\abs{f(\check{x}_k^{\tilde{\theta}_{k}^{(i)}})-f({x}_k^{\tilde{\theta}_{k}^{(i)}})}
&\leq 2\norm{f}_\infty \sup_{1\leq i\leq N}\mathbb{P}\left(\abs{x_k^{\tilde{\theta}_{k}^{(i)}}}>M_1\right)\\
&\leq 2\norm{f}_\infty \frac{\sup_{1\leq i\leq N}\mathbb{E}\abs{x_k^{\tilde{\theta}_{k}^{(i)}}}}{M_1}\\
&\leq 2\norm{f}_\infty \sqrt{\delta}\,.
\end{aligned}
\end{equation}
Substitute (\ref{lemma2_21}) into the first term on the right hand side of (\ref{lemma2_term2}) to obtain
\begin{equation}\label{lemma2_1}
\begin{aligned}
\sup_{1\leq i\leq N}{\left(\mathbb{E}\mathbf{1}_{A_\delta}\mathbf{1}_{B_{M_1}^c}\abs{f(\check{x}_k^{\tilde{\theta}_{k}^{(i)}})-f({x}_k^{\tilde{\theta}_{k}^{(i)}})}, {\Gamma}_{k-1}^{\tilde{\theta}_{k}^{(i)}}\right)}
&\leq \sup_{1\leq i\leq N}{\left(2\norm{f}_\infty \sqrt{\delta}, {\Gamma}_{k-1}^{\tilde{\theta}_{k}^{(i)}}\right)}\\
&\leq 2\norm{f}_\infty \sqrt{\delta}\,.
\end{aligned}
\end{equation}
\noindent
Next we consider the second term. Note that the random variables $\check{x}_k^{\tilde{\theta}_{k}^{(i)}}$ and ${x}_k^{\tilde{\theta}_{k}^{(i)}}$ restricted to  $B_{M_1}$ take values in a compact set. Hence the continuous function $f$ is also uniformly continuous on that set. Define $\epsilon=\sqrt{\delta}$, then there exists a $\delta_{\epsilon}$ such that $\abs{f(x)-f(y)}\leq \sqrt{\delta}$, for all $\abs{x-y}<\delta_\epsilon$. Define $\delta = \delta_{\epsilon}$. We obtain for the second term on the right hand side of (\ref{lemma2_term2}), using the definition of $A_\delta$,
\begin{equation*}
\sup_{1\leq i\leq N}\mathbb{E}\mathbf{1}_{A_{\delta_{\epsilon}}}\mathbf{1}_{B_{M_1}}\abs{f(\check{x}_k^{\tilde{\theta}_{k}^{(i)}})-f({x}_k^{\tilde{\theta}_{k}^{(i)}})}\leq \sqrt{\delta}\,,
\end{equation*}
which implies
\begin{equation}\label{lemma2_2}
\sup_{1\leq i\leq N}{\left(\mathbb{E}\mathbf{1}_{A_{\delta_{\epsilon}}}\mathbf{1}_{B_{M_1}}\abs{f(\check{x}_k^{\tilde{\theta}_{k}^{(i)}})-f({x}_k^{\tilde{\theta}_{k}^{(i)}})}, {\Gamma}_{k-1}^{\tilde{\theta}_{k}^{(i)}}\right)}\leq \sqrt{\delta}\,.
\end{equation}
\noindent For the last term on the right hand side of (\ref{lemma2_term2}), we apply again the Markov inequality, to obtain
\begin{equation*}
\begin{aligned}
\mathbb{E}\mathbf{1}_{A_{\delta_\epsilon}^c}\abs{f(\check{x}_k^{\tilde{\theta}_{k}^{(i)}})-f({x}_k^{\tilde{\theta}_{k}^{(i)}})}
&\leq 2\norm{f}_\infty \mathbb{P}(A_{\delta_\epsilon}^c)\\
&\leq 2\norm{f}_\infty \frac{\mathbb{E}\abs{\check{x}_k^{\tilde{\theta}_{k}^{(i)}}-{x}_k^{\tilde{\theta}_{k}^{(i)}}}}{\delta_\epsilon}\,.\\
\end{aligned}
\end{equation*}

\noindent We denote the $i$-th component of an $\mathbb{R}^d$-valued process $(x_k)_{k\in \mathbb{N}}$ by $(x^{(i)}_k)_{k\in \mathbb{N}}$. Furthermore, we temporarily suppress the dependence on $\tilde{\theta}_{k}^{(i)}$ in the notation.
Given $x_{k-1}$, according to Equations~(\ref{X}) and (\ref{Xhat}), we obtain
\begin{equation*}
\begin{aligned}
\mathbb{E}\abs{\check{x}_k^{(i)}-{x}_k^{(i)}}
& = \e^{-\alpha_i(t_k-t_{k-1})}\mathbb{E}\abs{\int_{t_{k-1}}^{t_{k}}\e^{\alpha_iu}\sum_{j=1}^{d}\Sigma_{ij}\left(\sqrt{x_{u}^{(1)}}-\sqrt{x_{k-1}^{(1)}}\right)\ud W_u^{(j)}}\\
&\leq \mathbb{E}\abs{\int_{t_{k-1}}^{t_{k}}\e^{\alpha_iu}\sum_{j=1}^{d}\Sigma_{ij}\left(\sqrt{x_{u}^{(1)}}-\sqrt{x_{k-1}^{(1)}}\right)\ud W^{(j)}_u}\,.
\end{aligned}
\end{equation*}
Define $\Sigma^{(i)}=\sqrt{\sum_{j=1}^{d}(\Sigma_{ij})^2}$. Using the Burkholder-Davis-Gundy inequality, we know there exists a constant $C$ which does not depend on $(x_t)_{t\geq0}$ such that, for every $i=1,\cdots,d$,
\begin{align*}
\mathbb{E}\abs{\int_{t_{k-1}}^{t_{k}}\e^{\alpha_iu}\sum_{j=1}^{d}\Sigma_{ij}\left(\sqrt{x_{u}^{(1)}}-\sqrt{x_{k-1}^{(1)}}\right)\,\ud W_u}
&\leq C\Sigma^{(i)} \mathbb{E}\left(\int_{t_{k-1}}^{t_{k}}\e^{2\alpha_iu}\left(\sqrt{x_u^{(1)}}-\sqrt{x_{k-1}^{(1)}}\right)^2\,\ud u\right)^{\frac{1}{2}}\\
&\leq C\Sigma^{(i)} \mathbb{E}\left(\int_{t_{k-1}}^{t_{k}}\e^{2\alpha_iu}\left(x_u^{(1)}+x_{k-1}^{(1)}\right)\,\ud u\right)^{\frac{1}{2}}\,.
\end{align*}
Using Jensen's inequality and Fubini's theorem, we get for the latter expectation
\begin{align*}
\lefteqn{\mathbb{E}\left(\int_{t_{k-1}}^{t_{k}}\e^{2\alpha_iu}\left(x_u^{(1)}+x_{k-1}^{(1)}\right)\,\ud u\right)^{\frac{1}{2}}}\\
&\qquad \leq  \left(\mathbb{E}\int_{t_{k-1}}^{t_{k}}\e^{2\alpha_iu}\left(x_u^{(1)}+x_{k-1}^{(1)}\right)\,\ud u\right)^{\frac{1}{2}}\\
&\qquad=  \left(\int_{t_{k-1}}^{t_{k}}\e^{2\alpha_iu}\mathbb{E}\left(x_u^{(1)}+x_{k-1}^{(1)}\right)\,\ud u\right)^{\frac{1}{2}}\\
&\qquad=  \left(\int_{t_{k-1}}^{t_{k}}\e^{2\alpha_i u}\left[(1+\e^{-\alpha_1(u-t_{k-1})})x_{k-1}^{(1)}+(1-e^{-\alpha_1(u-t_{k-1})})\beta_1\right] \ud u\right)^{\frac{1}{2}}\\
&\qquad= \left(\frac{x_{k-1}^{(1)}+\beta_1}{2\alpha_i}\left(\e^{2\alpha_i\delta}-1\right)e^{2\alpha_it_{k-1}}+\frac{x_{k-1}^{(1)}-\beta_1}{2\alpha_i-\alpha_1}\left(\e^{(2\alpha_i-\alpha_1)\delta}-1\right)e^{2\alpha_it_{k-1}}\right)^{\frac{1}{2}}.
\end{align*}
Note that $\e^{x}-1=O(x)$ if $x\rightarrow0$. Since the parameters are assumed to have a compact domain, we conclude there exists a constant $C_1$ which is independent of the parameters and such that
\begin{equation*}
\mathbb{E}\abs{\check{x}_k-{x}_k}\leq C_1x_{k-1}^{(1)}\sqrt{\delta}\,.
\end{equation*}
Hence, returning to previously used notation,
\begin{equation*}
\mathbb{E}\mathbf{1}_{A_{\delta_\epsilon}}^c\abs{f(\check{x}_k^{\tilde{\theta}_{k}^{(i)}})-f({x}_k^{\tilde{\theta}_{k}^{(i)}})}\leq \frac{C_1}{\delta_\epsilon}x_{k-1}^{{(1)},\tilde{\theta}_{k}^{(i)}}\sqrt{\delta}\,,
\end{equation*}
which implies
\begin{align}
\sup_{1\leq i\leq N}{\left(\mathbb{E}\mathbf{1}_{A_\delta^c}\abs{f(\check{x}_k^{\tilde{\theta}_{k}^{(i)}})-f({x}_k^{\tilde{\theta}_{k}^{(i)}})}, {\Gamma}_{k-1}^{\tilde{\theta}_{k}^{(i)}}\right)}
&\leq \frac{C_1}{\delta_\epsilon}\sqrt{\delta}\sup_{1\leq i\leq N}\mathbb{E}\abs{x_{k-1}^{{(1)},\tilde{\theta}_{k}^{(i)}}}\nonumber\\
&\leq \frac{C_1M}{\delta_\epsilon}\sqrt{\delta}\,.\label{lemma2_3}
\end{align}
\noindent Combining (\ref{lemma2_1}),(\ref{lemma2_2}) and (\ref{lemma2_3}) together with (\ref{lemma2_fisrt}), we prove the statement of the lemma.
\end{proof}

\noindent Lemma~\ref{lemma2} shows that the approximation error of $\hat{\pi}_k^{\tilde{\theta}_k^{(i)}}$ and $\hat{\Gamma}_{k-1}^{{\theta}_{k-1}^{(i)}}$ can be controlled in an appropriate manner, guaranteeing $\Sigma(\theta_{k-1})$ below a threshold value $V_N$. This allows us to run the outer layer in Algorithm~\ref{algoKPF} recursively, see step~\ref{once}. Adding Assumption~\ref{ass4}, we present in the following lemma the convergence of $\hat{\mu}_k^N$.

\begin{lem}\label{update}
Let the observation sequence $y_{1:k}$ be fixed and Assumptions~\ref{ass1}, \ref{ass2}, \ref{ass3} and \ref{ass4} hold. Then for any bounded and continuous function $f$, if
\begin{equation*}
\norm{(f,{\mu}_{k-1}^N)-(f,\mu_{k-1})}_p  \leq \frac{{c}_{1,{k-1}}}{\sqrt{N}}+{d}_{1,{k-1}}\sqrt{\delta}\,,
\end{equation*}
and
\begin{equation*}
\sup_{1\leq i\leq N} \abs{(f,\hat{\Gamma}_{k-1}^{\theta_{k-1}^{(i)}})-(f,{\Gamma}_{k-1}^{\theta_{k-1}^{(i)}})}\leq \frac{c_{2,{k-1}}}{\sqrt{N}}+d_{2,{k-1}}\sqrt{\delta}\,,
\end{equation*}
hold for some constants ${c}_{1,{k-1}},{d}_{1,{k-1}},c_{2,{k-1}}$ and $d_{2,{k-1}}$ which are independent of $N$ and $\delta$, then there exist constants $\hat{c}_{1,{k}}, \hat{d}_{1,k}, \tilde{c}_{2,{k}}$ and $\tilde{d}_{2,k}$ which are independent of $N$ and $\delta$  such that
\begin{align*}
\norm{(f,\hat{\mu}_{k}^N)-(f,\mu_{k})}_p & \leq \frac{\hat{c}_{1,{k}}}{\sqrt{N}}+\hat{d}_{1,{k}}\sqrt{\delta}\,,\\
\sup_{1\leq i\leq N} \abs{(f,\hat{\Gamma}_k^{\tilde{\theta}_{k}^{(i)}})-(f,{\Gamma}_{k}^{\tilde{\theta}_{k}^{(i)}})}&\leq \frac{\tilde{c}_{2,{k}}}{\sqrt{N}}+\tilde{d}_{2,{k}}\sqrt{\delta}\,.
\end{align*}
\end{lem}

\begin{proof}
First, we prove the convergence of the estimated normalized weights $\hat{w}_k^{\tilde{\theta}_k^{(i)}}$, which are used to prove the convergence of the measure $\hat{\mu}_{k}^N$.  Denote the unnormalized weights by $\hat{v}_k^{\tilde{\theta}_k^{(i)}}=((l^{\tilde{\theta}_k^{(i)}}_{y_k},\hat{\pi}_k^{\tilde{\theta}_k^{(i)}}),\hat{\Gamma}_{k-1}^{{\theta}_{k-1}^{(i)}})$, see \eqref{weights-for-theta}, and ${v}_k^{\tilde{\theta}_k^{(i)}}= ((l^{\tilde{\theta}_k^{(i)}}_{y_k},{\pi}_k^{\tilde{\theta}_k^{(i)}}),{\Gamma}_{k-1}^{\tilde{{\theta}}_{k}^{(i)}})$, see \eqref{eq:weights}. Using Lemma~\ref{lemma2}, we have
\begin{equation}\label{non-normalized_weights}
\begin{aligned}
\sup_{1\leq i\leq N}\abs{\hat{v}_k^{\tilde{\theta}_k^{(i)}}-{v}_k^{\tilde{\theta}_k^{(i)}}}
&=\sup_{1\leq i\leq N} \abs{((l_{\tilde{\theta}_k^{(i)}}^{y_k},\hat{\pi}_k^{\tilde{\theta}_k^{(i)}}),\hat{\Gamma}_{k-1}^{{\theta}_{k-1}^{(i)}})
 -((l_{\tilde{\theta}_k^{(i)}}^{y_k},{\pi}_k^{\tilde{\theta}_k^{(i)}}),{\Gamma}_{k-1}^{\tilde{\theta}_{k}^{(i)}})}\\
&\leq  \frac{\tilde{c}_{2,{k}}}{\sqrt{N}}+\tilde{d}_{2,{k}}\sqrt{\delta}\,,
\end{aligned}
\end{equation}
where $\tilde{c}_{2,{k}}$ and $\tilde{d}_{2,{k}}$ are constants which are independent of $N,\delta$ and $\theta$. By Assumption~\ref{ass4}, we obtain
\begin{equation*}
\begin{aligned}
\inf_{\theta\in D_\theta}\hat{v}_k^{\theta},\inf_{\theta\in D_\theta}{v}_k^{\theta}&>0\,,\\
\sup_{\theta\in D_\theta}\hat{v}_k^{\theta},\sup_{\theta\in D_\theta}{v}_k^{\theta}&<\infty\,.
\end{aligned}
\end{equation*}
Hence for the normalized weights, it follows that
\begin{equation}\label{inequal_weights}
\begin{aligned}
\sup_{1\leq i\leq N}\abs{\hat{w}_k^{\tilde{\theta}_k^{(i)}}-{w}_k^{\tilde{\theta}_k^{(i)}}}
&=\sup_{1\leq i\leq N}\abs{\frac{\hat{v}_k^{\tilde{\theta}_k^{(i)}}}{\sum_{i=1}^N \hat{v}_k^{\tilde{\theta}_k^{(i)}}}
-\frac{{v}_k^{\tilde{\theta}_k^{(i)}}}{\sum_{i=1}^N  {v}_k^{\tilde{\theta}_k^{(i)}}} }\\
&\leq \sup_{1\leq i\leq N}\abs{\frac{\hat{v}_k^{\tilde{\theta}_k^{(i)}}}{\sum_{i=1}^N \hat{v}_k^{\tilde{\theta}_k^{(i)}}}
- \frac{\hat{v}_k^{\tilde{\theta}_k^{(i)}}}{\sum_{i=1}^N {v}_k^{\tilde{\theta}_k^{(i)}}}}
+ \sup_{1\leq i\leq N}\abs{\frac{\hat{v}_k^{\tilde{\theta}_k^{(i)}}}{\sum_{i=1}^N {v}_k^{\tilde{\theta}_k^{(i)}}}
- \frac{{v}_k^{\tilde{\theta}_k^{(i)}}}{\sum_{i=1}^N  {v}_k^{\tilde{\theta}_k^{(i)}}}}\\
&\leq \sup_{1\leq i\leq N}\frac{\hat{v}_k^{\tilde{\theta}_k^{(i)}}}{(\sum_{i=1}^N {v}_k^{\tilde{\theta}_k^{(i)}})(\sum_{i=1}^N \hat{v}_k^{\tilde{\theta}_k^{(i)}})}\sum_{i=1}^N\abs{{v}_k^{\tilde{\theta}_k^{(i)}}-\hat{v}_k^{\tilde{\theta}_k^{(i)}}}\\
&\qquad + \sup_{1\leq i\leq N}\frac{1}{\sum_{i=1}^N \hat{v}_k^{\tilde{\theta}_k^{(i)}}}\abs{\hat{v}_k^{\tilde{\theta}_k^{(i)}}-{v}_k^{\tilde{\theta}_k^{(i)}}}\,.
\end{aligned}
\end{equation}
Observe that $\frac{\hat{v}_k^{\tilde{\theta}_k^{(i)}}}{\sum_{i=1}^N {v}_k^{\tilde{\theta}_k^{(i)}}}\leq 1$ and that $\sum_{i=1}^N \hat{v}_k^{\tilde{\theta}_k^{(i)}}$ is bounded from below by a constant times $N$.
Substituting Inequality (\ref{non-normalized_weights}) into Inequality (\ref{inequal_weights}), one obtains that there exist constants $\hat{c}_{2,{k}}$ and $\hat{d}_{2,{k}}$ such that
\begin{equation}\label{weights}
\sup_{1\leq i\leq N}\abs{\hat{w}_k^{\tilde{\theta}_k^{(i)}}-{w}_k^{\tilde{\theta}_k^{(i)}}} \leq  \frac{\hat{c}_{2,{k}}}{\sqrt{N}}+\hat{d}_{2,{k}}\sqrt{\delta}\,.
\end{equation}
\noindent Next we study the convergence of the measures $\hat{\mu}_{k}^N$. For simplification in the notation, we write $\hat{w}_k = \hat{w}_k^{\tilde{\theta}_k^{(i)}}$. Recalling that $w_k = \frac{p(y_k\mid y_{1:k-1},\theta)}{\int p(y_k\mid y_{1:k-1},\theta)d\theta}$, we get from Bayes' rule
\begin{align*}
(f,\mu_{k})        & = \int f(\theta)p(\theta\mid y_{1:k})d\theta\\
                   & =  \int f(\theta)\frac{p(y_k\mid y_{1:k-1},\theta)p(\theta\mid y_{1:k-1})}{\int p(y_k\mid y_{1:k-1},\theta)p(\theta\mid y_{1:k-1})d\theta}d\theta\\
                   & = \frac{(fw_k,\mu_{k-1})}{(w_k,\mu_{k-1})}\,,
\end{align*}
and from \eqref{eq:mu-tilde} and \eqref{eq:mu-hat} we get
\begin{align*}
(f,\hat{\mu}_{k}^N)& =\frac{(f\hat{w}_k,\tilde{\mu}_{k}^N)}{(\hat{w}_k,\tilde{\mu}_{k}^N)}\,.
\end{align*}
Since $w_k$ and $(w_k,\mu_{k-1})\geq 0$, using Assumption~\ref{ass4} and the triangle inequality, we get
\begin{align}
&\norm{(f,\hat{\mu}_{k}^N)-(f,\mu_{k})}_p\nonumber\\
&\quad =\norm{\frac{(f\hat{w}_k,\tilde{\mu}_{k}^N)}{(\hat{w}_k,\tilde{\mu}_{k}^N)}-\frac{(fw_k,\mu_{k-1})}{(w_k,\mu_{k-1})}}_p\nonumber\\
&\quad =\frac{1}{(w_k,\mu_{k-1})}\norm{\frac{(f\hat{w}_k,\tilde{\mu}_{k}^N)}{(\hat{w}_k,\tilde{\mu}_{k}^N)} (w_k,\mu_{k-1}) - (fw_k,\mu_{k-1})}_p\nonumber\\
&\quad \leq \frac{1}{(w_k,\mu_{k-1})}\left(\norm{\frac{(f\hat{w}_k,\tilde{\mu}_{k}^N)}{(\hat{w}_k,\tilde{\mu}_{k}^N)} (w_k,\mu_{k-1}) - (f\hat{w}_k,\tilde{\mu}_{k}^N)}_p+\norm{(f\hat{w}_k,\tilde{\mu}_{k}^N)-(fw_k,\mu_{k-1})}_p   \right)\nonumber\\
& \quad =  \frac{1}{(w_k,\mu_{k-1})}\left( \norm{\frac{(f\hat{w}_k,\tilde{\mu}_{k}^N)[(w_k,\mu_{k-1})-(\hat{w}_k,\tilde{\mu}_{k}^N)]}{(\hat{w}_k,\tilde{\mu}_{k}^N)}}_p
+ \norm{(f\hat{w}_k,\tilde{\mu}_{k}^N)-(fw_k,\mu_{k-1})}_p   \right)  \nonumber\\
&\quad \leq \frac{1}{(w_k,\mu_{k-1})}\left(\norm{f}_\infty\norm{(w_k,\mu_{k-1})-(\hat{w}_k,\tilde{\mu}_{k}^N)}_p+\norm{(f\hat{w}_k,\tilde{\mu}_{k}^N)-(fw_k,\mu_{k-1})}_p   \right)\,.\label{update_mu_inequal}
\end{align}
Hence, we need to find an upper bound for the two quantities $\norm{(w_k,\mu_{k-1})-(\hat{w}_k,\tilde{\mu}_{k}^N)}_p$ and $\norm{(f\hat{w}_k,\tilde{\mu}_{k}^N)-(fw_k,\mu_{k-1})}_p$.  Note that
\begin{equation}\label{eq:mutilde}
\norm{(f\hat{w}_k,\tilde{\mu}_{k}^N)-(fw_k,\mu_{k-1})}_p
\leq \norm{(fw_k,\mu_{k-1})-(fw_k,\tilde{\mu}_{k}^N)}_p+\norm{(fw_k,\tilde{\mu}_{k}^N)-(f\hat{w}_k,\tilde{\mu}_{k}^N)}_p.
\end{equation}
For the first term on the right hand side of \eqref{eq:mutilde}, it follows from Lemma~\ref{jittering} that
\begin{equation}\label{update_mu_term1}
\norm{(fw_k,\mu_{k-1})-(fw_k,\tilde{\mu}_{k}^N)}_p\leq\frac{\tilde{c}_{1,{k}}}{\sqrt{N}}+\tilde{d}_{1,{k}}\sqrt{\delta}\,.
\end{equation}
For the second term on the right hand side of \eqref{eq:mutilde}, we get using (\ref{weights})
\begin{equation*}
\begin{aligned}
\abs{(fw_k,\tilde{\mu}_{k}^N)-(f\hat{w}_k,\tilde{\mu}_{k}^N)}
&=\abs{\frac{1}{N}\sum_{i=1}^Nf(\tilde{\theta}_k^{(i)})\left(\hat{w}_k^{\tilde{\theta}_k^{(i)}}-{w}_k^{\tilde{\theta}_k^{(i)}}\right)}\\
&\leq \frac{\norm{f}_\infty}{N}\sum_{i=1}^N\abs{\hat{w}_k^{\tilde{\theta}_k^{(i)}}-{w}_k^{\tilde{\theta}_k^{(i)}}}\\
&\leq \frac{\tilde{c}_{2,{k}}}{\sqrt{N}}+\tilde{d}_{2,{k}}\sqrt{\delta}\,.
\end{aligned}
\end{equation*}
Inserting the latter together with (\ref{update_mu_term1}) in \eqref{eq:mutilde}, we obtain
\begin{equation}\label{update_mu_term12}
\norm{(f\hat{w}_k,\tilde{\mu}_{k}^N)-(fw_k,\mu_{k-1})}_p\leq \frac{{c}_{k}^\prime}{\sqrt{N}}+{d}_{k}^\prime\sqrt{\delta}\,,
\end{equation}
where ${c}_{k}^\prime$ and ${d}_{k}^\prime$ are constants independent of $N$ and $\delta$. Letting $f=1$ in \eqref{update_mu_term12} implies
\begin{equation}\label{update_mu_term11}
\norm{(w_k,\mu_{k-1})-(\hat{w}_k,\tilde{\mu}_{k}^N)}_p \leq \frac{{c}_{k}^\prime}{\sqrt{N}}+{d}_{{k}^\prime}\sqrt{\delta}\,.
\end{equation}
Therefore substituting (\ref{update_mu_term11}) and (\ref{update_mu_term12}) into the right hand side of Inequality (\ref{update_mu_inequal}), we obtain
\begin{equation*}
\norm{(f,\hat{\mu}_{k}^N)-(f,\mu_{k})}_p \leq \frac{\hat{c}_{1,{k}}}{\sqrt{N}}+\hat{d}_{1,{k}}\sqrt{\delta}\,,
\end{equation*}
where $\hat{c}_{1,{k}}=\frac{1+\norm{f}_\infty}{(w_k,\mu_{k-1})}{c}_{k}^\prime<\infty$ and $\hat{d}_{1,{k}}=\frac{1+\norm{f}_\infty}{(w_k,\mu_{k-1})}{d}_{k}^\prime<\infty$ are independent of $N$ and $\delta$ and the statement for $\hat{\mu}_k^N$ follows.

To prove the statement for $\hat{\Gamma}_k^{\tilde{\theta}_{k}^{(i)}}$, we compute using \eqref{posterior measure update} and \eqref{weights-for-theta}
\begin{align*}
\abs{(f,\hat{\Gamma}_k^{\tilde{\theta}_{k}^{(i)}})-(f,{\Gamma}_{k}^{\tilde{\theta}_{k}^{(i)}})}
&=\abs{ \frac{((fl^{\tilde{\theta}_{k}^{(i)}}_{y_k},\hat{\pi}_k^{\tilde{\theta}_{k}^{(i)}}),\hat{\Gamma}_{k-1}^{\theta_{k-1}^{(i)}}   )}{\hat{w}_{k}^{\tilde{\theta}_{k}^{(i)}}}
-\frac{((fl^{\tilde{\theta}_{k}^{(i)}}_{y_k},{\pi}_k^{\tilde{\theta}_{k}^{(i)}}),{\Gamma}_{k-1}^{\tilde{\theta}_{k}^{(i)}}   )}{{w}_{k}^{\tilde{\theta}_{k}^{(i)}}}}\\
&=\frac{1}{\hat{w}_{k}^{\tilde{\theta}_{k}^{(i)}} {w}_{k}^{\tilde{\theta}_{k}^{(i)}}} \abs{{w}_{k}^{\tilde{\theta}_{k}^{(i)}}((fl^{\tilde{\theta}_{k}^{(i)}}_{y_k},\hat{\pi}_k^{\tilde{\theta}_{k}^{(i)}}),\hat{\Gamma}_{k-1}^{\theta_{k-1}^{(i)}}   )-\hat{w}_{k}^{\tilde{\theta}_{k}^{(i)}}((fl^{\tilde{\theta}_{k}^{(i)}}_{y_k},{\pi}_k^{\tilde{\theta}_{k}^{(i)}}),{\Gamma}_{k-1}^{\tilde{\theta}_{k}^{(i)}}   )}\\
&\leq \frac{1}{\hat{w}_{k}^{\tilde{\theta}_{k}^{(i)}} {w}_{k}^{\tilde{\theta}_{k}^{(i)}}}
\left( \abs{{w}_{k}^{\tilde{\theta}_{k}^{(i)}}((fl^{\tilde{\theta}_{k}^{(i)}}_{y_k},\hat{\pi}_k^{\tilde{\theta}_{k}^{(i)}}),\hat{\Gamma}_{k-1}^{\theta_{k-1}^{(i)}}   )-\hat{w}_{k}^{\tilde{\theta}_{k}^{(i)}}((fl^{\tilde{\theta}_{k}^{(i)}}_{y_k},\hat{\pi}_k^{\tilde{\theta}_{k}^{(i)}}),\hat{\Gamma}_{k-1}^{\theta_{k-1}^{(i)}}   )}\right.\\
&\left.\quad +\abs{\hat{w}_{k}^{\tilde{\theta}_{k}^{(i)}}((fl^{\tilde{\theta}_{k}^{(i)}}_{y_k},\hat{\pi}_k^{\tilde{\theta}_{k}^{(i)}}),\hat{\Gamma}_{k-1}^{\theta_{k-1}^{(i)}}   )-\hat{w}_{k}^{\tilde{\theta}_{k}^{(i)}}((fl^{\tilde{\theta}_{k}^{(i)}}_{y_k},{\pi}_k^{\tilde{\theta}_{k}^{(i)}}),{\Gamma}_{k-1}^{\tilde{\theta}_{k}^{(i)}}   )}\right)\\
&\leq \frac{1}{\hat{w}_{k}^{\tilde{\theta}_{k}^{(i)}} {w}_{k}^{\tilde{\theta}_{k}^{(i)}}}
\left(\norm{f}_\infty\norm{l^{y_k}}_\infty\abs{{w}_{k}^{\tilde{\theta}_{k}^{(i)}}-\hat{w}_{k}^{\tilde{\theta}_{k}^{(i)}}}\right)\\
&\quad +\hat{w}_{k}^{\tilde{\theta}_{k}^{(i)}}\abs{((fl^{\tilde{\theta}_{k}^{(i)}}_{y_k},\hat{\pi}_k^{\tilde{\theta}_{k}^{(i)}}),\hat{\Gamma}_{k-1}^{\theta_{k-1}^{(i)}}   )-((fl^{\tilde{\theta}_{k}^{(i)}}_{y_k},{\pi}_k^{\tilde{\theta}_{k}^{(i)}}),{\Gamma}_{k-1}^{\tilde{\theta}_{k}^{(i)}}   )}\,.
\end{align*}
Since $\hat{w}_{k}^{\tilde{\theta}_{k}^{(i)}},{w}_{k}^{\tilde{\theta}_{k}^{(i)}}\geq \inf_{\theta\in D_\theta} l^{\theta}_{y_k}>0$, then using Equation~\eqref{weights} and Lemma~\ref{lemma2} (note that $l^{\tilde{\theta}_{k}^{(i)}}_{y_k}$ is bounded by Assumption~\ref{ass4}), we prove the statement of the Lemma.
\end{proof}

\subsection{Resampling}
In the following lemma we study the convergence of the measure $\mu_k^N$.
\begin{lem}\label{resampling}
Let the observation sequence $y_{1:k}$ be fixed, for bounded and continuous function $f$, if
\begin{equation*}
\begin{aligned}
\norm{(f,\hat{\mu}_{k}^N)-(f,\mu_{k})}_p & \leq \frac{\hat{c}_{1,{k}}}{\sqrt{N}}+\hat{d}_{1,{k}}\sqrt{\delta}\,,\\
\sup_{1\leq i\leq N} \abs{(f,\hat{\Gamma}_k^{\tilde{\theta}_{k}^{(i)}})-(f,{\Gamma}_{k}^{\tilde{\theta}_{k}^{(i)}})}&\leq \frac{\tilde{c}_{2,{k}}}{\sqrt{N}}+\tilde{d}_{2,{k}}\sqrt{\delta}\,,
\end{aligned}
\end{equation*}
holds for some constants $\hat{c}_{1,{k}},\hat{d}_{1,{k}},\tilde{c}_{2,{k}}$ and $\tilde{d}_{2,{k}}$ which are independent of $N$ and $\delta$, then there exist constants ${c}_{1,{k}},{d}_{1,{k}},{c}_{2,{k}}$ and ${d}_{2,{k}}$ which are independent of $N$ and $\delta$, such that
\begin{equation*}
\begin{aligned}
\norm{(f,{\mu}_{k}^N)-(f,\mu_{k})}_p & \leq \frac{{c}_{1,{k}}}{\sqrt{N}}+{d}_{1,{k}}\sqrt{\delta}\,,\\
\sup_{1\leq i\leq N} \abs{(f,\hat{\Gamma}_k^{{\theta}_{k}^{(i)}})-(f,{\Gamma}_{k}^{{\theta}_{k}^{(i)}})}&\leq \frac{{c}_{2,{k}}}{\sqrt{N}}+{d}_{2,{k}}\sqrt{\delta}\,.
\end{aligned}
\end{equation*}
\end{lem}

\begin{proof}
Note that in the resampling step the ${\theta}_{k}^{(i)}$ are resampled from the pool $\{\tilde{\theta}_{k}^{(i)},i=1,\cdots,N\}$. Hence it is trivial that
\begin{equation*}
\begin{aligned}
\sup_{1\leq i\leq N} \abs{(f,\hat{\Gamma}_k^{{\theta}_{k}^{(i)}})-(f,{\Gamma}_{k}^{{\theta}_{k}^{(i)}})}
&\leq \sup_{1\leq i\leq N} \abs{(f,\hat{\Gamma}_k^{\tilde{\theta}_{k}^{(i)}})-(f,{\Gamma}_{k}^{\tilde{\theta}_{k}^{(i)}})}
&\leq \frac{\tilde{c}_{2,{k}}}{\sqrt{N}}+\tilde{d}_{2,{k}}\sqrt{\delta}
& = \frac{{c}_{2,{k}}}{\sqrt{N}}+{d}_{2,{k}}\sqrt{\delta}\,,
\end{aligned}
\end{equation*}
where ${c}_{2,{k}}=\tilde{c}_{2,{k}}$ and ${d}_{2,{k}}=\tilde{d}_{2,{k}}$ are independent of $N$ and $\delta$.
Moreover, by triangle inequality, we have
\begin{equation}\label{resampling_12}
\norm{(f,{\mu}_{k}^N)-(f,\mu_{k})}_p \leq \norm{(f,{\mu}_{k}^N)-(f,\hat{\mu}_{k}^N)}_p + \norm{(f,\hat{\mu}_{k}^N)-(f,\mu_{k})}_p\,.
\end{equation}
For the second term on the right hand side of \eqref{resampling_12}, it follows from the conditions in the lemma, that
\begin{equation}\label{resampling_term2}
\norm{(f,\hat{\mu}_{k}^N)-(f,\mu_{k})}_p\leq \frac{\hat{c}_{1,{k}}}{\sqrt{N}}+\hat{d}_{1,{k}}\sqrt{\delta}\,.
\end{equation}
Note that $\{\theta_k^{(i)},i=1,\cdots,N\}$ are i.i.d.\ samples generated from $\hat{\mu}_k^N(\ud\theta)=\sum_{i=1}^N \hat{w}_k^{\tilde{\theta}_k^{(i)}}\delta_{\tilde{\theta}_k^{(i)}}(\ud\theta)$. Let $\tilde{\mathcal{G}}_k$ be the sigma-algebra generated by $ \{\theta_{1:k-1}^{(i)},\tilde{\theta}_{1:k}^{(i)}, i=1,\cdots,N\}$, then
\begin{equation*}
\begin{aligned}
\mathbb{E}[f(\theta_{k}^{(i)})\mid \tilde{\mathcal{G}}_k]& = \int f(\theta)\sum_{i=1}^N \hat{w}_k^{\tilde{\theta}_k^{(i)}}\delta_{\tilde{\theta}_k^{(i)}}(d\theta)\\
&= \sum_{i=1}^N \hat{w}_k^{\tilde{\theta}_k^{(i)}}f(\tilde{\theta}_k^{(i)})\\
&= (f,\hat{\mu}_k^N)\,.
\end{aligned}
\end{equation*}
Define $Z_k^{(i)} = f(\theta_k^{(i)})-(f,\hat{\mu}_k^N)=f(\theta_k^{(i)})-\mathbb{E}[f(\theta_{k}^{(i)})\mid \tilde{\mathcal{G}}_k]$.
The  $Z_k^{(i)}, i=1,\cdots,N$ are random variables with zero-mean and bounded by $2\norm{f}_\infty$ and have the property $\mathbb{E}[Z_k^{(i)}Z_k^{(j)}\mid \tilde{\mathcal{G}}_k]=0$ for $i\neq j$.
Let $p\geq 1$, and let's first additionally assume $p$ is an even integer. Then
\begin{equation*}
\begin{aligned}
\mathbb{E}\left[\abs{(f,\mu_k^N)-(f,\hat{\mu}_k^N)}^p \mid \tilde{\mathcal{G}}_k\right]
&=\mathbb{E}\left[\abs{\frac{1}{N}\sum_{i=1}^N f(\theta_k^{(i)})-(f,\hat{\mu}_k^N)}^p \mid \tilde{\mathcal{G}}_k\right]\\
&=\mathbb{E}\left[\abs{\frac{1}{N}\sum_{i=1}^N Z_k^{(i)}}^p \mid \tilde{\mathcal{G}}_k\right]  \\
&=\frac{1}{N^p}\mathbb{E}\left[\sum_{i_1=1}^N\cdots\sum_{i_p=1}^N Z_k^{(i_1)}\cdots Z_k^{(i_p)} \mid \tilde{\mathcal{G}}_k\right]\,.
\end{aligned}
\end{equation*}
Since $\mathbb{E}[Z_k^{(i)}\mid \tilde{\mathcal{G}}_k]=0$, there are at most $N^{\frac{p}{2}}$ non-zero contributions to
\begin{equation*}
\sum_{i_1=1}^N\cdots\sum_{i_p=1}^N \mathbb{E}\left[Z_k^{(i_1)}\cdots Z_k^{(i_p)} \mid \tilde{\mathcal{G}}_k\right]\,.
\end{equation*}
Hence
\begin{equation*}
\mathbb{E}\left[\sum_{i_1=1}^N\cdots\sum_{i_p=1}^N Z_k^{(i_1)}\cdots Z_k^{(i_p)} \mid \tilde{\mathcal{G}}_k\right]
\leq \frac{1}{N^{\frac{p}{2}}}2^p\norm{f}_\infty^p\,,
\end{equation*}
which implies
\begin{equation}\label{resampling_term1}
\norm{(f,\mu_k^N)-(f,\hat{\mu}_k^N)}_p\leq \frac{2\norm{f}_\infty}{\sqrt{N}}\,.
\end{equation}
Substituting Inequality (\ref{resampling_term1}) and (\ref{resampling_term2}) into Equation (\ref{resampling_12}) yields, for any even $p$, the result
\begin{equation}\label{resampling_even}
\norm{(f,{\mu}_{k}^N)-(f,\mu_{k})}_p\leq \frac{{c}_{1,{k}}}{\sqrt{N}}+{d}_{1,{k}}\sqrt{\delta}\,,
\end{equation}
where ${c}_{1,{k}} = 2\norm{f}_\infty+\hat{c}_{1,{k}}$ and ${d}_{1,{k}} =\hat{d}_{1,{k}}$.
For any real number $p\geq1$, we know there exist an even number $q>p$ such that $\eqref{resampling_even}$ holds for this number $q$. Hence the statement is proved.
\end{proof}

\subsection{Convergence of the Kalman Particle algorithm}\label{section:convergence}

In the following theorem we prove the convergence of our proposed algorithm.
\begin{thm}\label{theorem:convergence}
Suppose the function $f$ is bounded and continuous and Assumptions~\ref{ass1}, \ref{ass2}, \ref{ass3}, \ref{ass4} hold.
Let the sequence of the observation $y_{1:k}$ be fixed and the measures $\mu_k$ and $\mu_k^N$ resulting from Algorithm~\ref{algoKPF} be respectively as in \eqref{original-measure} and \eqref{approximating-measure},
where the model is Gaussian and linear or is of type (\ref{model_analysis}). Then it holds
\begin{equation*}
\norm{(f,\mu_k^N)-(f,\mu_k)}_p\leq \frac{c_{1,k}}{\sqrt{N}} + d_{1,{k}}\sqrt{\delta},
\end{equation*}
with constants $c_{1,k},d_{1,k}, 1\leq k\leq K$ independent of $N$ and $\delta$.
\end{thm}

\begin{proof}
We prove this theorem by induction. At time $t_0$, the parameter samples $\theta_{0}^{(i)}$, $i=1,\cdots,N$, are sampled from the initial measure $\mu_{0}$ and $\mu_{0}^N = \frac{1}{N}\sum_{i=1}^N \delta_{\theta_{0}^{(i)}}$. A well known result of Monte Carlo simulation, see for example Chapter I in \cite{glasserman2013monte}, implies that
\begin{equation*}
\norm{(f,\mu_{0})-(f,\mu_{0}^N)}\leq \frac{c_{1,t_0}}{\sqrt{N}} + d_{1,t_0}\sqrt{\delta},
\end{equation*}
for some constant $c_{1,t_0}$ and $d_{1,t_0}$ independent of $N$ and $\delta$. Moreover, one can just define the initial measure of $x_{0}$ by a Gaussian measure $\Gamma_{0}$ and define $\hat{\Gamma}_{0}^{\theta_{0}^{(i)}}= {\Gamma}_{0}^{\theta_{0}^{(i)}} = \Gamma_{0}$, for $i=1,\cdots,N$. Then it is trivial that
\begin{equation*}
\sup_{1\leq i\leq N}\norm{(f,\hat{\Gamma}_{0}^{\theta_{0}^{(i)}})-(f,{\Gamma}_{0}^{\theta_{0}^{(i)}})}\leq \frac{c_{2,t_0}}{\sqrt{N}}+d_{2,t_0}\sqrt{\delta}\,
\end{equation*}
holds for some constant $c_{2,t_0}$ and $d_{2,t_0}$ independent of $N$ and $\delta$ (actually $c_{2,t_0}=d_{2,t_0}=0$).
Assume that, at time ${k-1}$, the inequalities
\begin{equation*}
\norm{(f,{\mu}_{k-1}^N)-(f,\mu_{k-1})}_p  \leq \frac{{c}_{1,{k-1}}}{\sqrt{N}}+{d}_{1,{k-1}}\sqrt{\delta}
\end{equation*}
and
\begin{equation*}
\sup_{1\leq i\leq N} \abs{(f,\hat{\Gamma}_{k-1}^{\theta_{k-1}^{(i)}})-(f,{\Gamma}_{k-1}^{\theta_{k-1}^{(i)}})}\leq \frac{c_{2,{k-1}}}{\sqrt{N}}+d_{2,{k-1}}\sqrt{\delta}\,
\end{equation*}
hold for some constants ${c}_{1,{k-1}},{d}_{1,{k-1}},c_{2,{k-1}}$ and $d_{2,{k-1}}$ which are independent of $N$ and $\delta$. Then we can just successively apply Lemmas~\ref{update} and \ref{resampling} to obtain the statements of the lemma.
\end{proof}

\section{Numerical results on affine term structure models}\label{Numerical results}

In this section we illustrate our results by considering some specific examples of the affine class \eqref{model_analysis} to which we apply our algorithms. In particular we will consider the one factor Cox-Ingersoll-Ross (CIR) model, the two factor Hull-White model,  and the one factor Hull-White model with stochastic volatility. To test how the designed algorithm works on these models, the models are calibrated on  simulated data  using pre-determined parameters. We also compare the behavior of our algorithm to the one generated by the recursive nested particle filter (RNPF) for the CIR model. The comparison shows that in this case our algorithm outperforms the RNPF.

\subsection{One factor CIR model}\label{section:cir}

At first we consider the CIR model \eqref{eq:CIR-model}.
One important property of the CIR models is that the yield curves are always non-negative. The transition density of the CIR model has a non-central chi-square distribution, i.e.,
\begin{equation*}
x_{t_{k+1}}\mid x_{t_k} \sim c\chi_p^2(\lambda)\,,
\end{equation*}
where $p=\frac{4\alpha\beta}{\sigma^2}$ is the degree of freedom, $\lambda=x_k(4\alpha e^{-\alpha(t_{k+1}-t_{k})})/(\sigma^2(1-e^{-\alpha(t_{k+1}-t_k)}))$ is the non-centrality parameter and $c=\sigma^2(1-e^{-\alpha(t_{k+1}-t_k)})/4\alpha$.
We define the short rate process to be $r_t=x_t$. The analytical solution of the functions $\phi$ and $\psi$ defined in \eqref{eq:bond-explicit-formula} can be obtained by solving the ODE \eqref{riccati-equations2}, with $d=1$, $\gamma =-1$ and $c=0$.\\
\\
\noindent In the tests, the parameters are set to be $\alpha_1=0.45, \beta=0.001$ and $\tilde{\Sigma}=0.017$. Based on these parameters values, we generate daily data for yield curves with the times to maturity ranging from $1$ year to $30$ years. The time length is $T=2000$, i.e., the data set contains $2000$ days. Furthermore, we add white noise with variance $h = 1 \times 10^{-8}$ in the simulated zero rates.

We use our Kalman particle filter algorithm (Algorithm~\ref{algoKPF}, henceforth referred to as KPF) to calibrate the model parameters. The transition distribution of the CIR model is not Gaussian, but note that the CIR model fits the structure of (\ref{model_analysis}) with $\Sigma=0$. Hence we use the approximation (\ref{Xhat}) and apply the Kalman filter in the inner layer. In the outer layer, we set the number of particles to be $N=5000$ and the initial prior distribution of the parameters to be uniform, i.e.,
\begin{equation*}
\alpha \sim U(0,1),\ \ \ \beta \sim U(0,0.01), \ \ \ \mbox{and} \ \ \sigma \sim U(0,0.1)\,.
\end{equation*}
Moreover, the sampled parameters at each step are bounded by the boundary of the latter corresponding uniform distributions respectively. In the inner layer, the mean and variance of the initial prior distribution of $x_0$ are $0.005$ and $0.01$. The discounting factor $a$ is set to be $0.98$ and the variance boundary $V_N$ is set to be $\frac{1}{\sqrt{N^3}}$. We will also use the same $a$ and $V_N$ for other experiments later.

\begin{figure}[!htb]
	\centering
		\includegraphics[width=\linewidth]{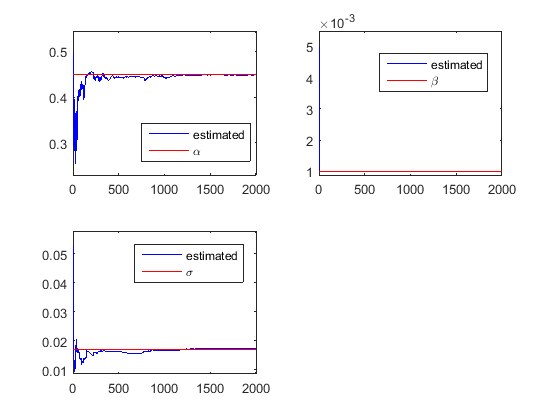}
\caption{Convergence of parameter estimates by KPF for the CIR model; the red lines represent the true values of the parameters.}
\label{fig:para_CIR}
\end{figure}

\begin{figure}[!htb]
	\centering
		\includegraphics[width=\linewidth]{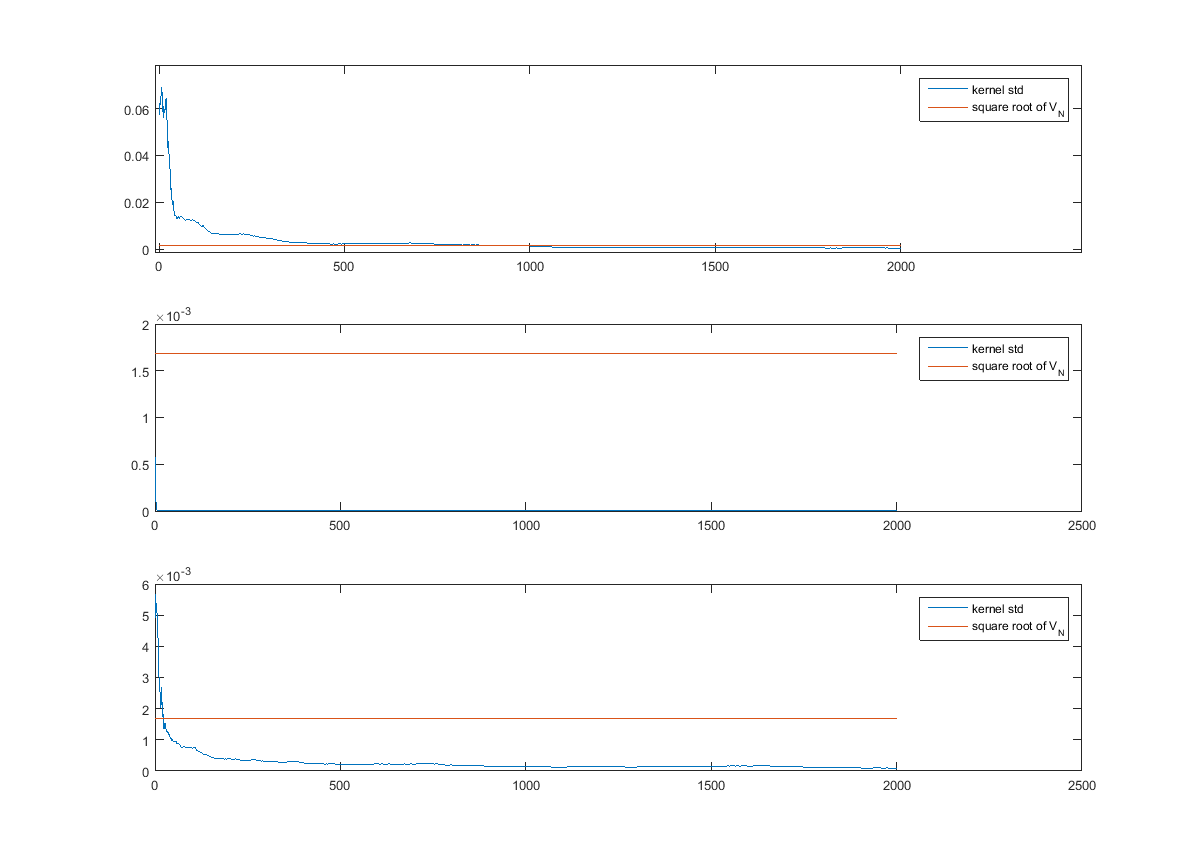}
\caption{Convergence of the standard deviation of the jittering kernel; the orange line represents the square root of switching level.}
\label{fig:variance_CIR}
\end{figure}
\noindent Figure~\ref{fig:para_CIR} shows that the estimated parameters converge over time. Figure~\ref{fig:variance_CIR} shows the convergence of the standard deviation of the jittering kernel for the three parameters. It is seen that after 437 steps the variance is smaller than the required level.
\medskip\\
We also implemented the recursive nested particle filter (RNPF) for  comparison. The sample size in the two layers are set to be 1000 and 300 respectively. Moreover, we simulated 18000 more days of data, hence in total we obtain 20000 days of data. For the rest we use the same settings  as in the previous example, i.e.\ the same initial sample distributions for the parameter generation, the same boundary on the parameter samples in the outer layer, and the variance of the jittering kernel is also set the same as $V_N$.
\begin{figure}[!htb]
	\centering
		\includegraphics[width=\linewidth]{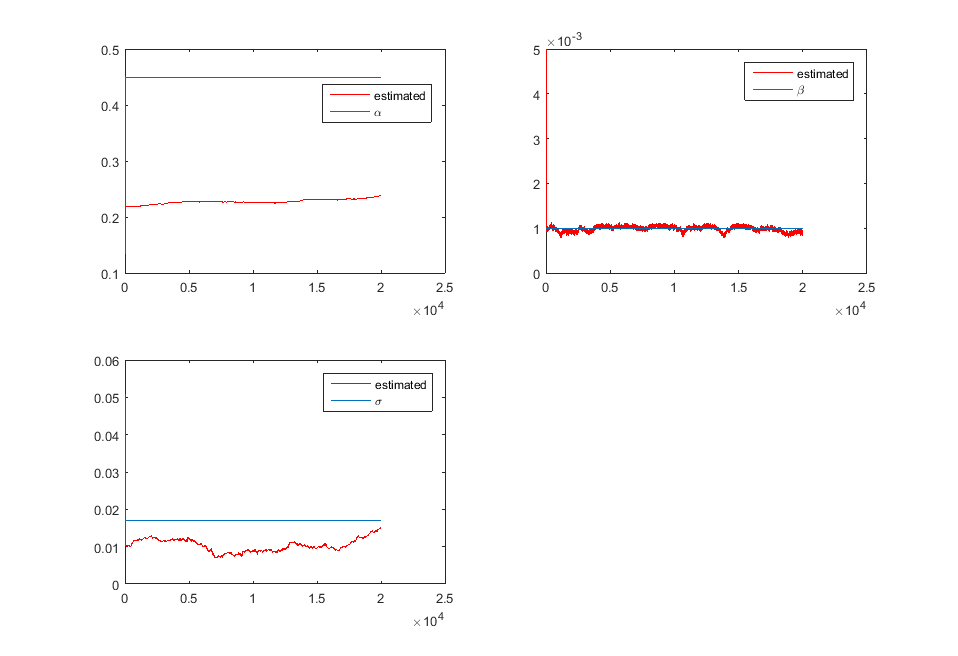}
\caption{Behavior of parameter estimates by RNPF for the CIR model.}
\label{fig:para_CIR_RNPF}
\end{figure}
\noindent Figure~\ref{fig:para_CIR_RNPF} shows how the behavior of the estimated parameters over time. One observes that even after 20000 time steps the RNPF algorithm estimates of $\alpha$  and $\sigma$ don't reach the correct parameter values.

\subsection{Two factor Hull-White model}\label{section:hw}

In this subsection we consider the two factor Hull-White model \eqref{eq:HW}.
The short rate process $r_t$ is given by $r_t = x_t^{(1)} + x_t^{(2)}$.
The analytical solution of the functions $\phi$ and $\psi$ defined in \eqref{eq:bond-explicit-formula} can be obtained by solving the ODE \eqref{riccati-equations1}, with $d=2$,
$\gamma = (-1, -1)^\top$ and $c=0$. Furthermore we set
$$\alpha_{11} = 0.03, \,\alpha_{22} = 0.23,\, \Sigma_{11}=0.02,\, \Sigma_{12}=0, \,\Sigma_{21} = 0.02\rho, \,\Sigma_{22} = 0.02 \sqrt{1-\rho^2}, \,\rho = -0.5\,.$$
\noindent Similar as in the previous example, the yield curve data are simulated based on these parameters and then a white noise process is added on the simulated data. The variance of the white noise is $6\times 10^{-7}$. The times to maturity of the yield curves range from $1$ year up to $30$ years. The time step of the data is set to be daily and the time length is $2000$.
\medskip\\
Using these noisy simulated data, we use our Kalman particle algorithm~\ref{algoKPF} to calibrate the model parameters. The Hull-White model fits the structure of (\ref{model_analysis}) with $\Sigma=0$. Since the Hull-White model is a Gaussian model, the Kalman filter in the inner layer gives an optimal filter.  The number of particles at each step is $N = 2000$. For the initialization of the parameter samples, the initial prior distribution of parameters is chosen to be uniform, namely:
\begin{equation*}
\alpha_1,\alpha_2 \sim U(0,0.4),\ \ \ \sigma_1,\sigma_2 \sim U(0,0.1), \ \ \ \mbox{and} \ \ \rho \sim U(-0.8,-0.3)\,.
\end{equation*}
The initial prior distribution of the state $x_k$ is chosen to be Gaussian with mean $0$ and variance $\mathrm{diag}([0.1,0.1])$. Figure~\ref{fig:para_HW} shows the how the estimated parameters converge over time. One can observe that the convergence is very fast and accurate.

\begin{figure}[h!]
	\centering
		\includegraphics[width=\linewidth]{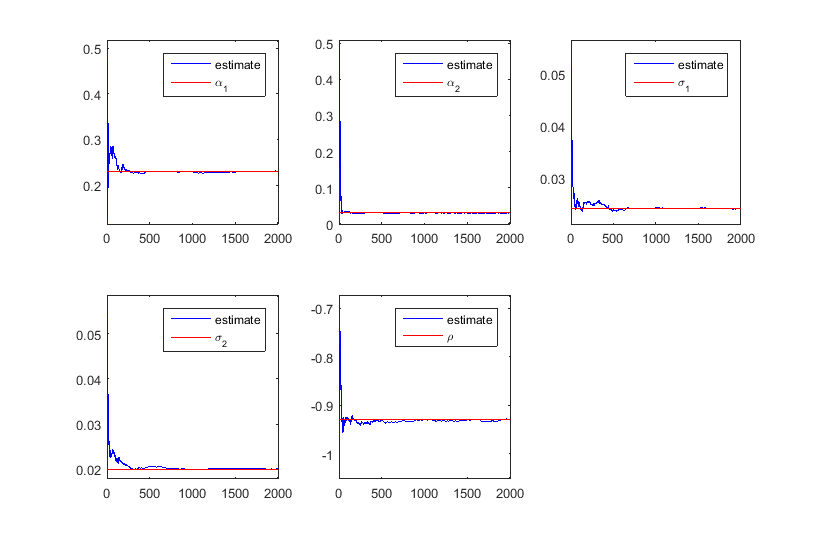}
\caption{Convergence of parameter estimates by KPF for the Hull-White model.}
\label{fig:para_HW}
\end{figure}

\subsection{Hull-White model with stochastic volatility}

We consider the following stochastic volatility model,
\begin{equation*}
\begin{aligned}
\ud V_t &= \alpha_1(0.1-V_t)\,\ud t+\sigma_1\sqrt{V_t}\,\ud W_t^{(1)},\\
\ud X_t &= \alpha_2(\beta-X_t)\,\ud t+\sigma_2\sqrt{V_t}\left(\rho \,\ud W_t^{(1)}+\sqrt{1-\rho^2}\,\ud W_t^{(2)}\right),
\end{aligned}
\end{equation*}
which is model~\eqref{eq:HSV} in a notation that is more suitable for the purposes of this section.
The process $V_t$ presents the fluctuation of the volatility of the system. Note that the long term mean of the process $V_t$ is fixed at the known constant $0.1$, otherwise the model would be over-parametrized, i.e.\ by scaling the volatility process and the parameters $\sigma_1, \sigma_2$ one can obtain an equivalent model. The transition density of this stochastic volatility model is not analytically available. To tackle this issue, we apply the same approximation as in the CIR model test of Section~\ref{section:cir}, namely, we approximate the stochastic diffusion by constant diffusion between the time steps, see (\ref{Xhat}). Under this approximation, the transition density of the model is Gaussian and the mean and variance can be theoretically computed. The short rate is defined by $r_t=X_t$ and hence the yield curve can be computed by $R_t(\tau) = \phi(\tau)+\psi(\tau)^\top \tilde{X}_t$, with $\tilde{X}_t = [V_t,X_t]^\top$. The functions $\phi$, $\psi$ are the solutions to the Riccati equations \eqref{riccati-equations2}, with $d=2$, $\gamma = (0,-1)^\top$ and $c=0$. The solutions to these latter equations are not known in closed form. We introduce an efficient numerical algorithm to compute these functions, the detailed algorithm is in the Appendices \ref{ric:numerical}.
The parameters of this model are set to be $[\alpha_1,\alpha_2,\beta,\sigma_1,\sigma_2,\rho]=[0.1,0.3,0.03,0.3,0.07,-0.5]$. The variance of the white noise is $10^{-8}$. The times to maturity of the yield curves range from $1$ year up to $20$ years. The time step of the data is set to be daily and the time length is $2000$. In the outer layer, we sample $N=2000$ particles and in the inner layer, we use the Kalman filter. The initial prior distributions of the parameters are uniform,
\begin{equation*}
\alpha_1,\alpha_2 \sim U(0,1),\ \ \beta \sim U(0,0.1), \ \ \sigma_1\sim U(0,0.8), \ \ \ \sigma_2\sim U(0,0.2) \mbox{  and} \ \ \ \rho\sim U(-1,1).
\end{equation*}
The mean and variance of the initial prior distribution of $\tilde{X}_t$ are $[0.1,0]$ and $\mathrm{diag}([0.01,0.01])$. Figure~\ref{fig:para_HW_SV} shows that also for a model with stochastic volatility the parameter estimates quickly converge.

\begin{figure}[h!]
	\centering
		\includegraphics[width=1.1\linewidth]{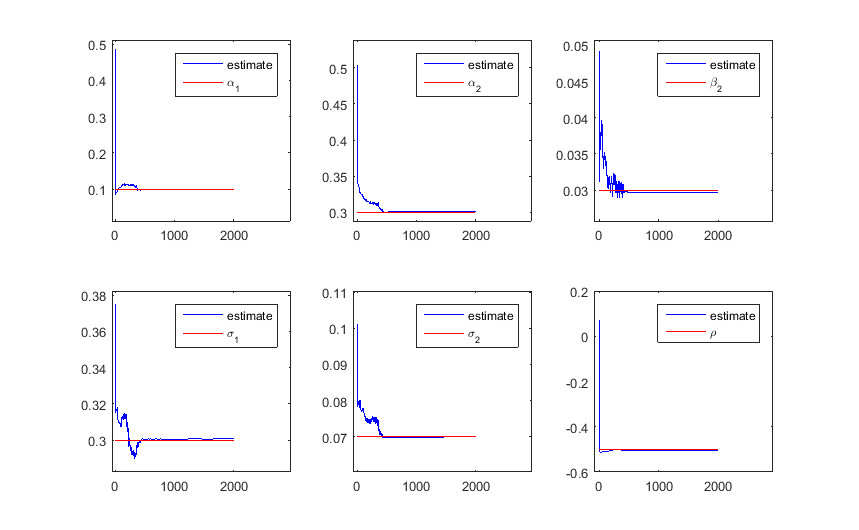}
		\caption{Convergence of parameter estimates by KPF for the Hull-White model with stochastic volatility.}
\label{fig:para_HW_SV}
\end{figure}
\noindent

\subsection{One factor CIR model with jump parameters}

This experiment can be seen as an extension of the experiment on the CIR model calibration of Section~\ref{section:cir}. In this experiment, we assume the parameters have a jump at time $T=2001$, from $[\alpha,\beta,\sigma] = [0.45,0.001,0.017]$ to $[\alpha,\beta,\sigma] = [0.55,0.0015,0.023]$. We simulate the new data from time point $T=2001$ to $T=4000$ based on the new parameters and the settings for the other parameters are the same as in Section~\ref{section:cir}. To identify the parameter change, we set $b=0.1$. In the experiment we use Algorithm~\ref{algoKPFJ}. Figure~\ref{fig:para_CIR_c} shows that in this study the KPF algorithm for models with time-varying parameters is able to track a sudden change in the  parameter values and quickly stabilizes at the new values.

\begin{figure}[h!]
	\centering
		\includegraphics[width=1.1\linewidth]{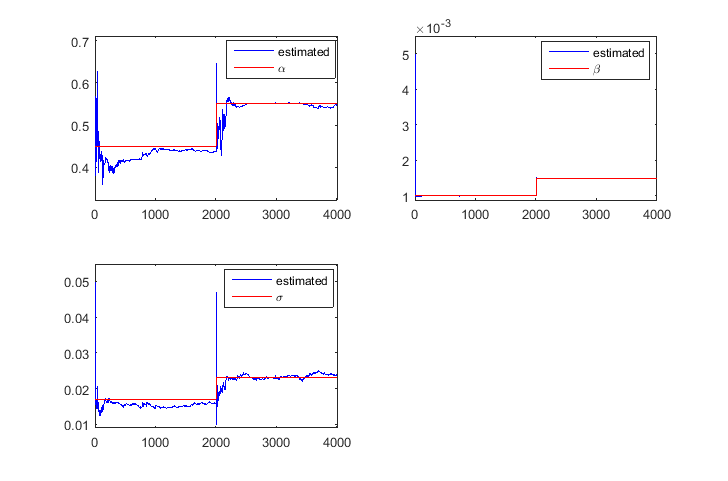}
\caption{Parameter estimates by KPF for the CIR model with a sudden jump.}
\label{fig:para_CIR_c}
\end{figure}

\section{Conclusion}\label{conclusion}

In this paper we have introduced a semi-recursive algorithm combining the Kalman filter and the particle filter with a two layers structure. In the outer layer the dynamic Gaussian kernel is implemented to sample the parameter particles. Moreover, the Kalman filter is applied the inner layer to estimate the posterior distribution of the state variables given the parameters sampled in the outer layer. These two changes provide faster convergence and reduce the computational time comparable to the RNPF methodology. The theoretical contribution of this paper is the convergence analysis of the proposed algorithm. We proved that, under regularity assumptions and given a certain model structure, the posterior distribution of the parameters and the state variables converge to the actual distribution in $L_p$ with rate $\mathcal{O}(N^{-\frac{1}{2}}+\delta^{\frac{1}{2}})$. The theoretical result is complemented by numerical results for several affine term structure models with static parameters or jump parameters. Although our numerical illustrations are for term structure models, the Kalman particle algorithm can also be applied to many other models.

\section*{Acknowledgement}
The authors thanks the contribution from Dr. Peter den Iseger and the support from ABN AMRO Bank N.V.

\appendix

\section{Affine processes}

Affine processes are continuous-time Markov processes characterised by the fact that their characteristic function depends in an exponentially affine way on the initial state vector of the process.
From Theorem 2.7 in \cite{DDW}, we know that the model of type \eqref{model_analysis} is an $\mathbb{R}_+^p\times \mathbb{R}^q$-valued affine process given $\Sigma$ or $\tilde{\Sigma}$ is zero and the admissibility of the parameters of this model.

\subsection{Admissibility of the parameters}\label{app-section:admissibility}

Let $\Gamma = \Sigma \Sigma^\top$, $\tilde{\Gamma}=\tilde{\Sigma}\tilde{\Sigma}^\top$, $I = \{1,\cdots, p\}$ and $J=\{p+1,\cdots,p+q\}$. Here below, we introduce the admissibility of the parameters of \eqref{model_analysis} when $\tilde{\Sigma}=0$,
\begin{itemize}
\item  $\Gamma_{II} = 0$ \,,
\item $A\beta \in \mathbb{R}_+^p\times \mathbb{R}^q$\,,
\item  $A_{IJ} =0$\,,
\item $A_{II}$ has nonpositive off-diagonal elements\,.
\end{itemize}
When $\Sigma =0$, then the admissibility of the parameters of \eqref{model_analysis} reads
\begin{itemize}
\item  $\Gamma= 0$,  if $p=0$\,,
\item $\Gamma_{kl} = \Gamma_{lk} = 0$, for $k\in I/\{1\}$, for all $1\leq l \leq d$\,,
\item $A\beta \in \mathbb{R}_+^p\times \mathbb{R}^q$\,,
\item  $A_{IJ} =0$\,,
\item $A_{II}$ has nonpositive off-diagonal elements\,.
\end{itemize}
This latter conditions on the parameters ensure that the process of \eqref{model_analysis} remains in the state space $\mathbb{R}_+^p\times \mathbb{R}^q$.

\subsection{Riccati equations}\label{sec:Riccati-equations}

Let $u \in \mathbb{C}^d$, $(\phi(\cdot, u), \psi(\cdot, u)): [0,T] \rightarrow \mathbb{C}\times \mathbb{C}^d$ and $(\tilde{\phi}(\cdot, u), \tilde{\psi}(\cdot, u)): [0,T] \rightarrow \mathbb{C}\times \mathbb{C}^d$ be $C^1$-functions.
We introduce the following generalised Riccati equations.\medskip\\
{\bf Generalised Riccati equations (1).}
\begin{equation}\label{riccati-equations1}
\begin{aligned}
\partial_t \phi(t,u) &= \frac{1}{2}\psi_J^\top(t,u)\Gamma_{JJ} \psi_J(t,u)) +(A\beta)^\top \psi(t,u) -c\,,\\
 \phi(0,u) &= 0\,, \\
\partial_t \psi_i(t, u) &=- A_i^\top \psi(t, u) -\gamma_i\,, \quad 1\leq i \leq d\,, \\
\partial_t \psi_J(t, u) &= -A_{JJ} \psi(t, u) - \gamma_J\,,\\
\psi(0, u) &= u \,,
\end{aligned}
\end{equation}
{\bf Generalised Riccati equations (2).}
\begin{equation}\label{riccati-equations2}
\begin{aligned}
\partial_t \tilde{\phi}(t,u) &= (A\beta)^\top \tilde{\psi}(t,u) -c\,,\\
\tilde{\phi}(0,u) &= 0\,, \\
 \partial_t \tilde{\psi}_1(t, u) &= \frac{1}{2}\tilde{\psi}^\top(t, u) \tilde{\Gamma}_1^\top \tilde{\psi}(t, u) - A_1^\top \tilde{\psi}(t, u)-\gamma_1\,,  \\
\partial_t \tilde{\psi}_i(t, u) &=- A_i^\top \tilde{\psi}(t, u) -\gamma_i\,, \quad 2\leq i \leq d\,, \\
\partial_t \tilde{\psi}_J(t, u) &= -A^\top_{JJ} \tilde{\psi}_J(t, u) - \gamma_J\,,\\
\tilde{\psi}(0, u) &= u \,.
\end{aligned}
\end{equation}
The aim in the following theorem is to compute the zero coupon bond price $P(t,T)$ introduced in \eqref{eq:bond-explicit-formula}.
For a proof, we refer to Theorem 3.1 in \cite{KM}.
\begin{thm} \label{thm:discounting}
Let $\tau>0$ and $(x_t)_{t\geq 0}$ be as in \eqref{model_analysis} with $\tilde{\Sigma}=0$. Then the following statements are equivalent
\begin{enumerate}
\item $\mathbb{E}[\e^{-\int_0^\tau r(s) \, \ud s}] < \infty$\,, for some $x \in \mathbb{R}_+^p\times \mathbb{R}^q$.
\item
There exists a unique solution $(\phi, \psi)$ on $[0,\tau]$ to the generalised Riccati equations \eqref{riccati-equations1} with initial data $u=0$.
\end{enumerate}
In any of the above cases, it holds for all $0\leq t\leq T\leq \tau$ and for all $x \in \mathbb{R}_+^p\times \mathbb{R}^q$,
$$\mathbb{E}[\e^{-\int_t^	T r(s) \, \ud s} \mid \mathcal{F}_t] = \e^{-\phi(T-t,0) -\psi(T-t,0) x(t)}\,.$$
The above statements remain true when $\Sigma=0$ if we replace $(\phi, \psi)$ by $(\tilde{\phi}, \tilde{\psi})$, the solution to~\eqref{riccati-equations2}.
\end{thm}

\subsection{Numerical Solution to Riccati equations}\label{ric:numerical}
For many Riccati equations, it is hard (or even impossible) to calculate a closed-form solution, especially in high dimensional cases. So a numerical approach is needed. In general, the Riccati equations for $\phi(t,u)$ and $\psi(t,u)$ are given by
\begin{equation}
\begin{aligned}
\partial_t\phi(t,u)&=\frac{1}{2}\psi(t,u)^\top a\psi(t,u)+b^\top\psi(t,u)-c,\\
\phi(0,u)&=0;\\
\partial_t\psi_i(t,u)&=\frac{1}{2}\psi(t,u)^\top\alpha_i\psi(t,u)+\beta_i^\top\psi(t,u)-\gamma_i,\\
\psi(0,u)&=u.
\end{aligned}
\label{ric}
\end{equation}
with known parameters $a,b,c, \alpha_i, \beta_i $ and $\gamma_i, i=1,\cdots,d$. \\
\\
We use a Taylor series to approximate the solution $(\phi,\psi)$. In order to do so, first we need to determine the coefficients in Taylor expansion.
\begin{prop}
Suppose $(\phi(t,u),\psi(t,u))$ is the solution of (\ref{ric}). Given the value of $u$,and assume the Taylor expansions of $(\phi(t,u),\psi(t,u))$ are given by $\phi(t,u)=\sum_{k=0}^\infty C_k(u)t^k<\infty$, and $\psi_i(t,u)=\sum_{k=0}^\infty D_k^i(u)t^k<\infty$, then we have the following recursion for the coefficients:
\begin{equation*}
\begin{aligned}
C_0(u)&=0,\\
C_1(u)&=\frac{1}{2}u^\top au+b^\top u-c,\\
C_{k+1}(u)&=\frac{1}{1+k}\left(\frac{1}{2}\sum_{n=0}^kD_n^\top(u)aD_{k-n}(u)+b^\top D_ku\right), k\geq 2,  \\
D_0^i(u)&=u_i,\\
D_1^i(u)&=\frac{1}{2}u^\top\alpha_iu+\beta_i^\top u-\gamma_i,\\
D_{k+1}^i(u)&=\frac{1}{1+k}\left(\frac{1}{2}\sum_{n=0}^kD_n^\top(u)\alpha_iD_{k-n}(u)+\beta_i^\top D_k(u)\right), k\geq2,
\end{aligned}
\label{Ric}
\end{equation*}
where $D_k(u)=(D_k^1(u),\cdots, D_k^d(u))^\top, k=0,\cdots,n$.
\end{prop}
\begin{proof}
Suppose $\phi(t,u)=\sum_{k=0}^\infty C_k(u)t^k, \psi_i(t,u)=\sum_{k=0}^\infty D_k^i(u)t^k$, let $t=0$, we obtain $C_0(u)=0, D_0^i(u)=u_i$. Taking the derivative of $\psi_i(t,u)$ w.r.t.\ $t$,
\begin{align}
\partial_t\psi_i(t,u)&=\sum_{k=1}^\infty D_{k}^i(u)kt^{k-1} \nonumber\\
&=\sum_{k=0}^\infty D_{k+1}^i(u)(k+1)t^{k}.
\label{left}
\end{align}

\noindent On the other hand, according to (\ref{ric}),

\begin{align}
\partial_t\psi_i(t,u)&=\frac{1}{2}\psi(t,u)^\top\alpha_i\psi(t,u)+\beta_i^\top\psi(t,u)-\gamma_i \nonumber\\				
&=\frac{1}{2}\sum_{l,r=1}^d \psi_l(t,u)\alpha_i(l,r)\psi_r(t,u)+\sum_{s=1}^dB(s,i)\psi_s(t,u)-\gamma_i \nonumber\\
&=\frac{1}{2}\sum_{l,r=1}^d (\sum_{k=0}^\infty D_k^l(u)t^k)\alpha_i(l,r)(\sum_{k=0}^\infty D_k^r(u)t^k)+\sum_{s=1}^dB(s,i)(\sum_{k=0}^\infty D_k^s(u)t^k)-\gamma_i \nonumber\\
&=\frac{1}{2}\sum_{l,r=1}^d\alpha_i(l,r)\sum_{k=0}^\infty(\sum_{m=0}^kD_m^l(u)D_{k-m}^r(u))t^k+\sum_{s=1}^dB(s,i)(\sum_{k=0}^\infty D_k^s(u)t^k)-\gamma_i \nonumber\\
&=\sum_{k=0}^\infty\left(\frac{1}{2}\sum_{l,r=1}^d\sum_{m=0}^kD_m^l(u)\alpha_i(l,r)D_{k-m}^r(u)+\sum_{s=1}^dB(s,i)D_k^s(u)t^k\right)t^k-\gamma_i \nonumber\\
&=\sum_{k=0}^\infty\left(\frac{1}{2}\sum_{m=0}^kD_m^\top(u)\alpha_iD_{k-m}(u)+\beta_i^TD_k(u)\right)t^k-\gamma_i \nonumber\\
&=(\frac{1}{2}u^\top\alpha_iu+\beta_i^\top u-\gamma_i)+\sum_{k=1}^\infty\left(\frac{1}{2}\sum_{m=0}^kD_m^\top(u)\alpha_iD_{k-m}(u)+\beta_i^\top D_k(u)\right)t^k.
\label{right}
\end{align}

\noindent Comparing the Taylor coefficients in (\ref{left}) and (\ref{right}), we obtain
\begin{eqnarray*}
D_{1}(u)&=&\frac{1}{2}u^\top\alpha_iu+\beta_i^\top u-\gamma_i,\\
D_{k+1}^i(u)&=&\frac{1}{1+k}\left(\frac{1}{2}\sum_{n=0}^kD_n^\top(u)\alpha_iD_{k-n}(u)+\beta_i^\top D_k(u)\right).
\end{eqnarray*}
Similarly, we also obtain
\begin{eqnarray*}
C_1(u)&=&\frac{1}{2}u^\top au+b^\top u-c,\\
C_{k+1}(u)&=&\frac{1}{1+k}\left(\frac{1}{2}\sum_{n=0}^kD_n^\top(u)aD_{k-n}(u)+b^\top D_ku\right).
\end{eqnarray*}
\end{proof}
\noindent
This proposition allows us to approximate the $(\phi(t,u),\psi(t,u))$ by
\begin{eqnarray*}
\phi(t,u)&\approx&\sum_{k=0}^NC_k(u)t^k, \\
\psi_i(t,u)&\approx&\sum_{k=0}^ND_k^i(u)t^k.
\end{eqnarray*}
The approximation errors are of the form $\sum_{k=N+1}^\infty A_k(u)t^k$. The approximation is accurate and converges quickly if $t\approx0$. For $t\gg0$, we divide the time interval into several subintervals which are small enough to make the approximation accurate.\\
\\
Choose time steps $\Delta_i>0, i=1,\cdots,n$ such that $T-t=\Delta_1+\cdots+\Delta_n$, then by the tower property,
\begin{eqnarray*}
e^{\phi(T-t,u)+\psi(T-t,u)^\top X(t)}&=&\mathbb{E}[e^{u^\top X(T)}\mid \mathcal{F}_t]\\
&=&\mathbb{E}[\mathbb{E}[e^{u^\top X(T)}\mid\mathcal{F}_{T-\Delta_1}]\mid \mathcal{F}_t]\\
&=&e^{\phi(\Delta_1,u)}\mathbb{E}[e^{\psi(\Delta_1,u)^\top X(T-\Delta_1)}\mid\mathcal{F}_t]\\
&=&e^{\phi(\Delta_1,u)}\mathbb{E}[\mathbb{E}[e^{\psi(\Delta_1,u)^\top X(T-\Delta_1)}\mid\mathcal{F}_{T-\Delta_1-\Delta_2}]\mid\mathcal{F}_t]\\
&=&e^{\phi(\Delta_1,u)+\phi(\Delta_2,\psi(\Delta_1,u))}\mathbb{E}[e^{\psi(\Delta_2,\psi(\Delta_1,u))^\top X(T-\Delta_1-\Delta_2)}\mid\mathcal{F}_t]\\
&\vdots&\\
&=&e^{\phi(\Delta_1,u_0)+\phi(\Delta_2,u_1)+\cdots+\phi(\Delta_n,u_{n-1})}e^{\psi(\Delta_n,u_{n-1})^\top X(t)},
\end{eqnarray*}
where $u_{i+1}=\psi(\Delta_{i+1},u_i), u_0=u$.\\
\\
Comparing the two extreme sides of the latter equation, we obtain
\begin{eqnarray*}
\phi(T-t,u)&=&\sum_{i=1}^n\phi(\Delta_i,u_{i-1}),\\
\psi(T-t,u)&=&\psi(\Delta_n,u_{n-1}).
\end{eqnarray*}
In practice, we can set the approximation error level to be $\varepsilon$. If at each step we choose $\Delta=(\varepsilon A_N(u))^{\frac{1}{N}}$, then the last term in the Taylor expansion is $A_N(u)*\Delta^N=\varepsilon$. Hence we can control the approximation at the level $\varepsilon$.

\begin{remark}
The values of the functions $\phi(t,u),\psi(t,u)$ might go to infinity for some value of $t$ and $u$ . In these situations, the Taylor expansion approximation doesn't work. However, in financial application, we assume these cases do not exist since in finance we always assume the moments of the underlying process exist.
\end{remark}

\begin{example}
\begin{figure}[!htb]
    \centering
    \includegraphics[width=0.4\textwidth]{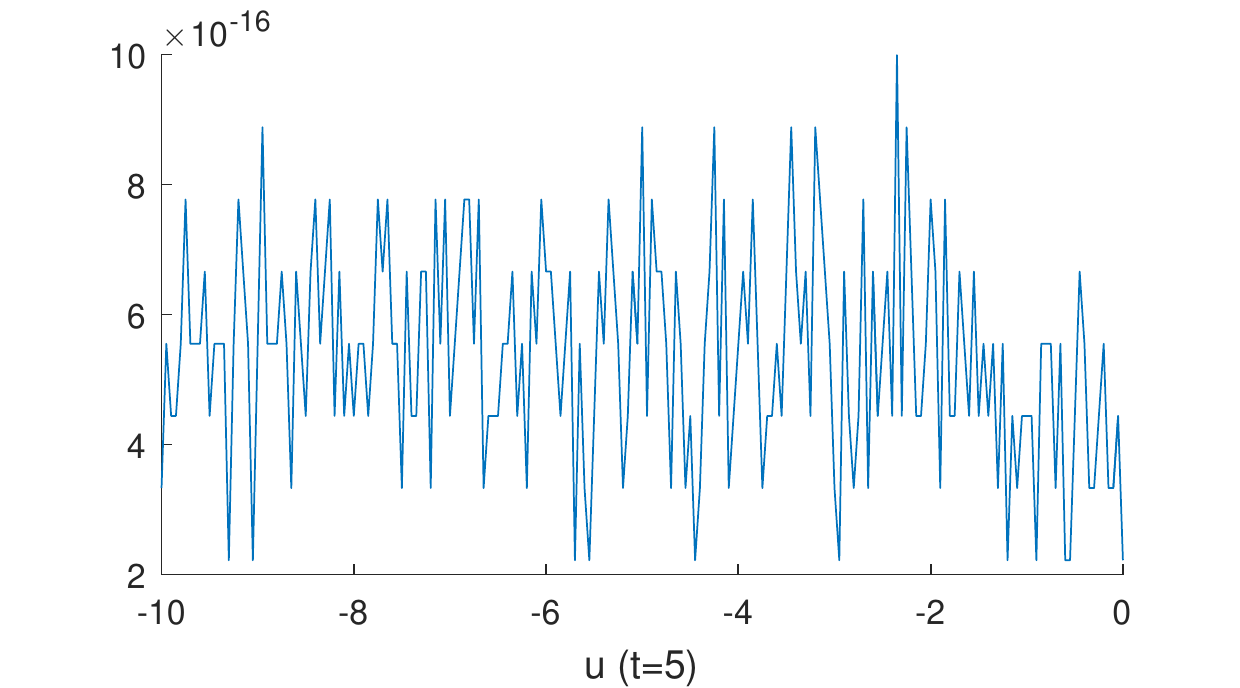}
    \includegraphics[width=0.4\textwidth]{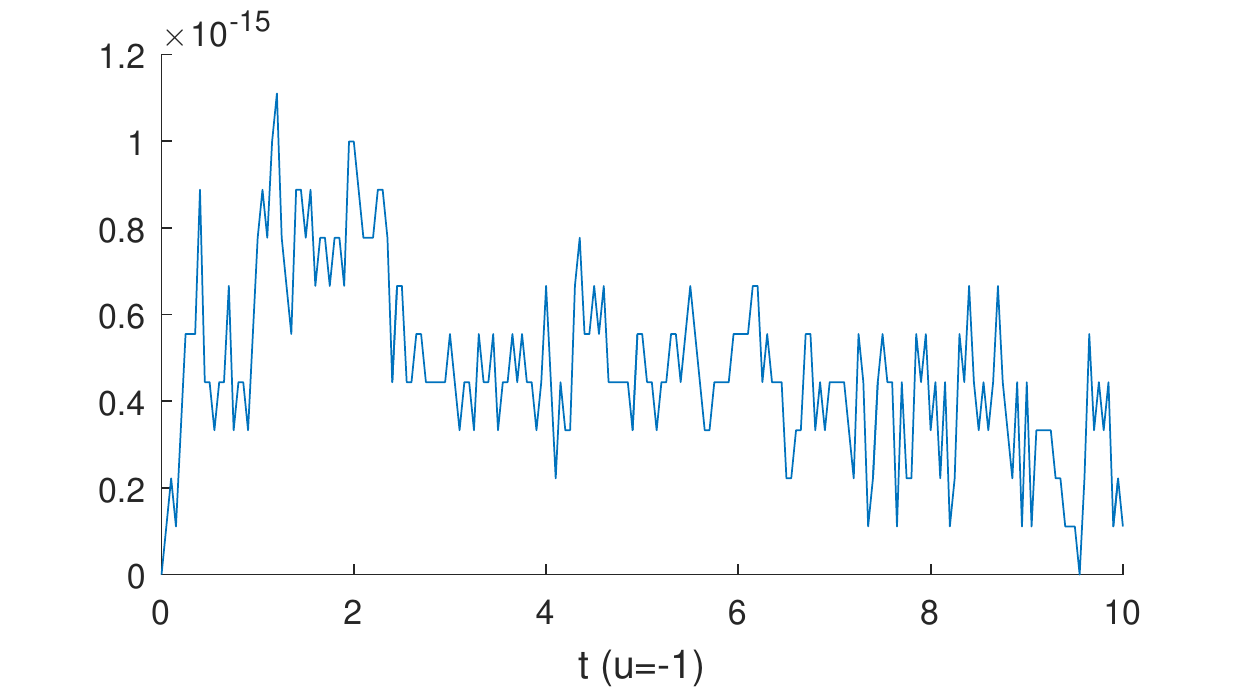}
    \caption{Numerical errors of the approximate solutions}
    \label{ric:error}
\end{figure}
Consider \eqref{ric} with $d=1$ and the admissible parameters $\alpha=\gamma=1, \beta=-1$. The ODE for $\psi$ is
\begin{eqnarray*}
\partial_t\psi(t,u)&=&\psi(t,u)^2-\psi(t,u)-1,\\
\psi(0,u)&=&u.
\end{eqnarray*}
The unique closed form solution to this equation is given by, see \cite[Eq.(10.47)]{Filipovic},
\begin{eqnarray*}
\psi(t,u)=\frac{2(e^{\sqrt{5}t}-1)-((\sqrt{5}-1)e^{\sqrt{5}t}+\sqrt{5}+1)u}{(\sqrt{5}+1)e^{\sqrt{5}t}+\sqrt{5}-1-2(e^{\sqrt{5}t}-1)u}.
\end{eqnarray*}
For the numerical approximation, we choose the Taylor expansion order $N=10$ and the tolerance of the error $\epsilon=10^{-16}$.
The plots in Figure~\ref{ric:error} show the numerical errors for different $t$ and $u$.
\end{example}

\bibliographystyle{plain}

\end{document}